\documentclass[%
 reprint,
 superscriptaddress,
amsmath,amssymb,
aps,
pra,
floatfix,
]{revtex4-2}
\usepackage{graphicx}
\usepackage{amsthm,amsfonts}
\usepackage{dcolumn}
\usepackage{bm}
\usepackage{xcolor}
\usepackage{booktabs}
\usepackage{floatrow}
\usepackage[caption=false]{subfig}
\usepackage[utf8]{inputenc}
\usepackage{pgfplots}
\pgfplotsset{compat=1.3}
\usepackage{braket}
\usepackage{tikz}
\usetikzlibrary{arrows.meta}

\newtheorem{theorem}{Theorem}

\newtheorem{lemma}{Lemma}
\theoremstyle{definition}
\newtheorem{definition}{Definition}

\usepackage{algorithm}
\usepackage[noend]{algpseudocode}

\usepackage{tabularx}
\makeatletter
\newcommand{\multiline}[1]{%
  \begin{tabularx}{\dimexpr\linewidth-\ALG@thistlm}[t]{@{}X@{}}
    #1
  \end{tabularx}
}
\makeatother
\usepackage{xcolor}

\usepackage{hyperref}
\hypersetup{
    colorlinks,
    linkcolor={red!50!black},
    citecolor={green!50!black},
    urlcolor={blue!50!black}
}

\usepackage{mathtools}
\DeclarePairedDelimiter{\ceil}{\lceil}{\rceil}
\DeclareMathOperator{\im}{im}
\DeclareMathOperator{\St}{St}

\DeclareMathOperator{\usetop}{\uparrow}
\DeclareMathOperator{\dsetop}{\downarrow}
\DeclareMathOperator{\ddiam}{diam}
\newcommand{\BigO}[1]{\ensuremath{O\left(#1\right)}}

\newcommand{\reg}[1]{\mathcal{R}\left(#1\right)}
\newcommand{\tri}[2]{\Delta_{#1}\left(#2\right)}
\newcommand{\dist}[2]{d\left(#1,#2\right)}

\newcommand{\diam}[1]{\ddiam(#1)}
\newcommand{\dmond}[1]{\lozenge\left(#1\right)}
\newcommand{\future}[1]{\usetop(#1)}
\newcommand{\past}[1]{\dsetop(#1)}
\newcommand{\cpath}[2]{(#1\updownarrow#2)}
\newcommand{\rest}[1]{|_{#1}}

\newcommand\change[2][]{#2}
\newcommand\todo[1]{#1}

\usepackage[capitalize]{cleveref}

\begin{document}

\title{Cellular automaton decoders for topological quantum codes \\with noisy measurements and beyond}
\author{Michael Vasmer}
\email{mvasmer@perimeterinstitute.ca}
\affiliation{Department of Physics and Astronomy, University College London, Gower Street, London, WC1E 6BT}
\affiliation{Perimeter Institute for Theoretical Physics, Waterloo, ON N2L 2Y5, Canada}
\affiliation{Institute for Quantum Computing, University of Waterloo, Waterloo, ON N2L 3G1, Canada}
\author{Dan E. Browne}
\affiliation{Department of Physics and Astronomy, University College London, Gower Street, London, WC1E 6BT}
\author{Aleksander Kubica}
\email{akubica@perimeterinstitute.ca}
\affiliation{Perimeter Institute for Theoretical Physics, Waterloo, ON N2L 2Y5, Canada}
\affiliation{Institute for Quantum Computing, University of Waterloo, Waterloo, ON N2L 3G1, Canada}
\date{\today}

\begin{abstract}
    We propose an error correction procedure based on a cellular automaton, the sweep rule, which is applicable to a broad range of codes beyond topological quantum codes.
    For simplicity, however, we focus on the three-dimensional (3D) toric code on the rhombic dodecahedral lattice with boundaries and prove that the resulting local decoder has a non-zero error threshold.
    We also numerically benchmark the performance of the decoder in the setting with measurement errors using various noise models.
    We find that this error correction procedure is remarkably robust against measurement errors and is also essentially insensitive to the details of the lattice and noise model. Our work constitutes a step towards finding simple and high-performance decoding strategies for a wide range of quantum low-density parity-check codes. 
\end{abstract}

\maketitle

\section{Introduction \label{sec:intro}}

Developing and optimizing decoders for quantum error-correcting codes is an essential task on the road towards building a fault-tolerant quantum computer~\cite{Shor1996, Preskill1998, Campbell2017}.
A decoder is a classical algorithm that outputs a correction operator, given an error syndrome, i.e.\ a list of measurement outcomes of parity-check operators. 
The error threshold of a decoder tells us the maximum error rate that the code (and hence an architecture based on the code) can tolerate. 
Moreover, decoder performance has a direct impact on the resource requirements of fault-tolerant quantum computation~\cite{Fowler2012}.
In addition, the runtime of a decoder has a large bearing on the clock speed of a quantum computer and may be the most significant bottleneck in some architectures~\cite{Terhal2015, Das2020}.

Here, we focus on decoders for CSS stabilizer codes~\cite{Gottesman1997, Calderbank1996}, in particular topological quantum codes~\cite{Kitaev2003, Dennis2002, Bombin2013book, Brown2014}, which have desirable properties such as high error-correction thresholds and low-weight stabilizer generators.
Recently, there has been renewed interest in $d$-dimensional topological codes, where $d\geq 3$, because they have more powerful logical gates than their two-dimensional counterparts~\cite{Bombin2007, Bombin2013, Kubica2015a, Kubica2015, Webster2019, Vasmer2019, Bravyi2013a, Pastawski2014, Jochym-OConnor2018} and are naturally suited to networked ~\cite{Barrett2005,Nickerson2014,Monroe2014,Kalb2017} and ballistic linear optical architectures~\cite{Kieling2007,Raussendorf2006,Gimeno-Segovia2015,Rudolph2017,Nickerson2018}. 
In addition, it has recently been proposed that one could utilize the power of three-dimensional (3D) topological codes using a two-dimensional layer of active qubits~\cite{Bombin2018,Browneaay4929}.

\floatsetup[figure]{style=plain,subcapbesideposition=top}
\begin{figure}[!ht]
    \centering
    \includegraphics[width=0.9\textwidth]{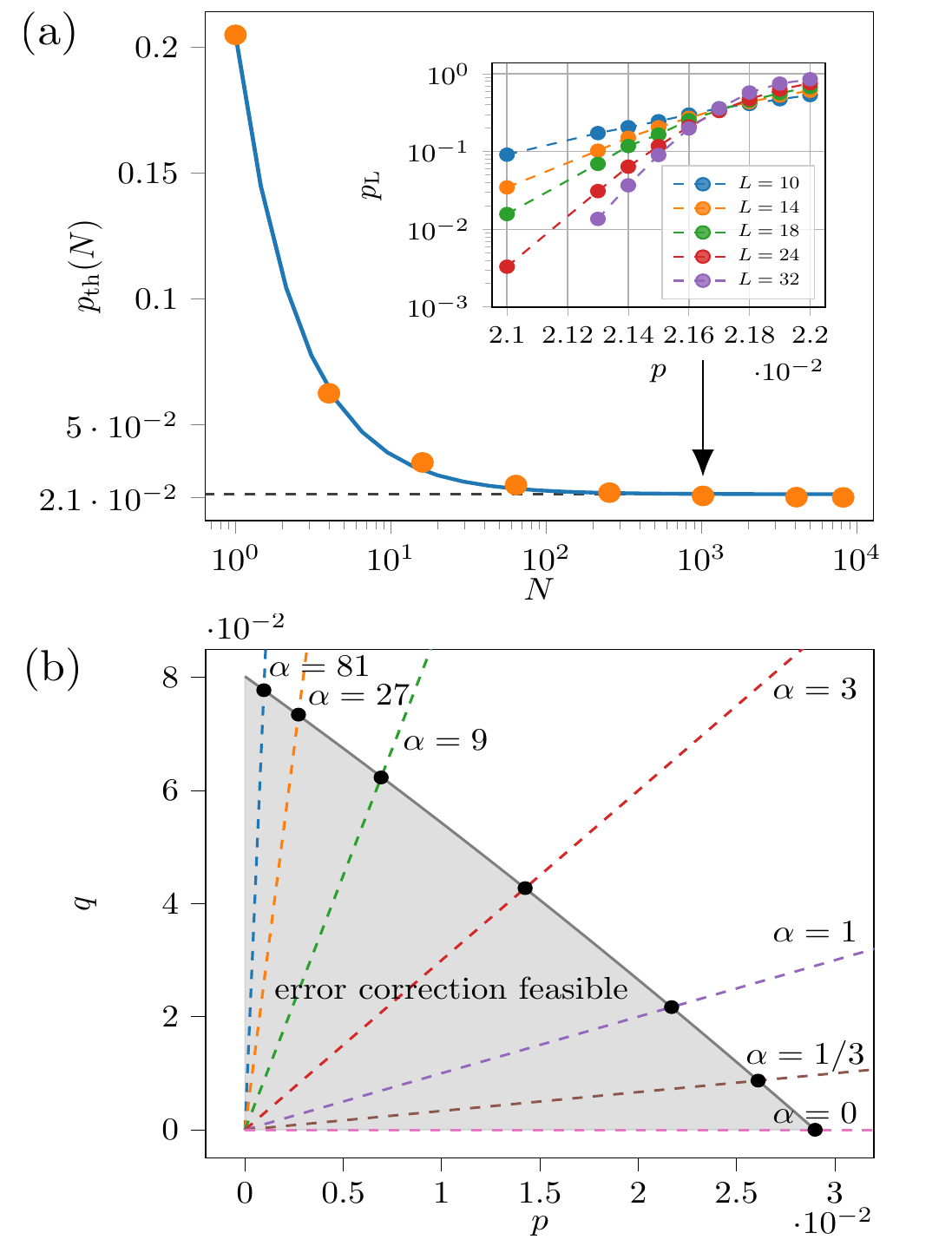}
    \caption{Numerical error threshold estimates for the sweep decoder applied to the toric code on the rhombic dodecahedral lattice. In (a), we plot the error threshold $p_{\mathrm{th}}(N)$ as function of the number of error-correction cycles $N$, for an error model with equal phase-flip ($p$) and measurement error ($q$) probabilities ($\alpha=q/p=1$). The inset shows the data for $N=2^{10}$, where we use $10^4$ Monte Carlo samples for each point. Using the ansatz in \cref{eq:sus-thresh}, we estimate the sustainable threshold to be $p_{\mathrm{sus}}\approx 2.1\%$. 
    In (b), we plot $p_{\mathrm{sus}}$ for error models with different values of $\alpha$, where we approximate $p_{\mathrm{sus}}\approx p_{\mathrm{th}}(2^{10})$. 
    % The shaded region shows where error correction is feasible.
    }
    \label{fig:rhombic-toric}
\end{figure}

Cellular-automaton (CA) decoders for topological codes~\cite{Harrington2004, Pastawski2011, Herold2015, Herold2017, Breuckmann2017, Dauphinais2017, Kubica2019} are particularly attractive because they are local: at each vertex of the lattice we compute a correction using a simple rule that processes syndrome information in the neighbourhood of the vertex.
This is in contrast to more complicated decoders such as the minimum-weight perfect matching algorithm~\cite{Dennis2002}, which requires global processing of the entire syndrome to compute a correction.
Moreover, CA decoders have another advantage: they exhibit single-shot error correction~\cite{Bombin2015, Campbell2019}, i.e.\ it is not necessary to repeat stabilizer measurements to compensate for the effect of measurement errors.

In our work we study the recently proposed sweep decoder~\cite{Kubica2019}, a cellular-automaton decoder based on the sweep rule.
% First, we adapt the sweep rule to an abstract setting of codes without geometric locality, which opens up the possibility of using the sweep decoder for generic low-density parity-check codes.
First, we adapt the sweep rule to an abstract setting of codes without geometric locality, which opens up the possibility of using the sweep decoder for \todo{certain low-density parity-check (LDPC) codes beyond topological codes.}
Second, we show how the sweep decoder can be used to decode \todo{phase-flip errors in} the 3D toric code on the rhombic dodecahedral lattice with boundaries, and we prove that it has a non-zero error threshold in this case. 
We remark that the original sweep decoder only works for the toric code defined on lattices without boundaries.
Third, we numerically simulate the performance of the decoder in the setting with measurement errors and further optimize its performance.
We use an independent and identically distributed (iid) error model with phase-flip probability $p$ and measurement error probability $q$.
We observe an error threshold of ${\sim}2.1\%$ when $q=p$, an error threshold of ${\sim}2.9\%$ when $q=0$ and an error threshold of ${\sim}8\%$ when $p\rightarrow 0^+$; see \cref{fig:rhombic-toric}. 
\todo{We note that the sweep rule decoder cannot be used to decode bit-flip errors in the 3D toric code, as their (point-like) syndromes do not have the necessary structure. However, 2D toric code decoders can be used for this purpose, e.g.~\cite{Dennis2002,Wang2011,Fowler2014,Delfosse2017,Bravyi2013,Duclos-Cianci2010,Duclos-Cianci2014}.}

We report that the sweep decoder has an impressive robustness against measurement errors and in general performs well in terms of an error-correction threshold.
To compare the sweep decoder with previous work~\cite{Kubica2018a,Breuckmann2018,Duivenvoorden2019,Aloshious2019}, we look at the 3D toric code on the cubic lattice, as decoding the 3D toric code on the rhombic dodecahedral lattice has not been studied before; see \cref{tab:pth_compare}.

\begin{table}[ht!]
    \centering
    \begin{ruledtabular}
	\begin{tabular}{lll}
        Decoder & $q=0$ & $q=p$ \\
        \hline
		Neural Network ~\cite{Breuckmann2018} & 17.5\% & N/A \\
		Renormalization Group (RG)~\cite{Duivenvoorden2019} & $17.2\%$ & $7.3\%$ \\
		Toom's Rule~~\cite{Kubica2018a} & $14.5\%$ & N/A \\
		Erasure Mapping (EM)~\cite{Aloshious2019} & $12.2\%$ & N/A \\
		\textbf{Sweep (with direction change)} & \textbf{15.5\%} & \textbf{1.7\%} \\
    \end{tabular}
    \end{ruledtabular}
	\caption{Comparison of the error thresholds of (cubic) toric code decoders against phase-flip noise with ($q=p$) and without ($q=0$) measurement errors. Our results are shown in bold.}
	\label{tab:pth_compare}
\end{table}

The remainder of this article is structured as follows.
In \cref{sec:results}, we start by presenting how the sweep rule can be used in an abstract setting of codes without geometric locality.
Then, we outline a proof of the non-zero error threshold of the sweep decoder for the 3D toric code on the rhombic dodecahedral lattice with boundaries.
We also present numerical simulations of the performance of the decoder in the setting with measurement errors for lattices with and without boundaries.
Next, in \cref{sec:discussion}, we discuss the applicability of the sweep decoder and suggest directions for further research.
Finally, in \cref{sec:methods}, we prove the properties of the sweep rule in the abstract setting, and analyze the case of lattices with boundaries.

\section{Results}
\label{sec:results}

We start this section by adapting the sweep rule~\cite{Kubica2019} to the setting of causal codes, which go beyond topological quantum codes.
Then, we focus on the 3D toric code on the rhombic dodecahedral lattice. 
We first analyze the case of the infinite lattice, followed by the case of lattices with boundaries.
We finish by presenting numerical simulations of the performance of various optimized versions of the sweep decoder.

\subsection{Sweep rule for causal codes}
\label{subsec:abstractSR}

\change[We now present a generalization of the sweep rule, which allows to decode CSS stabilizer codes~\cite{Gottesman1997,Calderbank1996,Steane1996} beyond topological quantum codes.]{}
Recall that a stabilizer code is CSS iff its stabilizer group can be generated by operators that consist exclusively of Pauli $X$ or Pauli $Z$ operators.
Let $\mathcal Q$ denote the set of physical qubits of the code and $\mathcal{S}$ be the set of all $X$ stabilizer generators, which are measured.
We refer to the stabilizers returning $-1$ outcome as the $X$-type syndrome.
The $X$-type syndrome constitutes the classical data needed to correct Pauli $Z$ errors.
In what follows, we focus on correcting $Z$ errors as $X$ errors are handled analogously.

We start by introducing a partially ordered set $V$ with a binary relation $\preceq$ over its elements.
We refer to the elements of $V$ as locations.
Given a subset of locations $U\subseteq V$, we say that a location $w\in V$ is an upper bound of $U$ and write $U \preceq w$ iff $u \preceq w$ for all $u \in U$; a lower bound of $U$ is defined similarly.
The supremum of $U$, denoted by $\sup U$, is the least upper bound of $U$, i.e., $\sup U \preceq w$ for each upper bound $w$ of $U$.
Similarly, the infimum $\inf U$ is the greatest lower bound of $U$.
We also define the future and past of $w\in V$ to be, respectively
\begin{eqnarray}
\future w &=& \{ W \subseteq V : w \preceq W\},\\
\past w &=& \{ W \subseteq V : W \preceq w\}.
\end{eqnarray}
We define the causal diamond
\begin{equation}
\dmond U = \future{\inf U} \cap \past{\sup U}
\end{equation}
of $U$ as the intersection of the future of $\inf U$ and the past of $\sup U$.
Lastly, for any $\mathcal A \subseteq 2^V$, where $2^V$ is the power set of $V$, we define the restriction of $\mathcal A$ to the location $v\in V$ as follows
\begin{equation}
\mathcal A\rest v = \{ A \in \mathcal A : A \ni v \}.
\label{eq:restrict}
\end{equation}
For notational convenience, we use the shorthands $\bigcup\mathcal A = \bigcup_{A\in\mathcal A} A$ and $\sup\mathcal A = \sup \bigcup\mathcal A$.

Let $C_{\mathcal Q}$ and $C_{\mathcal S}$ be $\mathbb{F}_2$-linear vector spaces with the sets of qubits $\mathcal Q$ and $X$ stabilizer generators $\mathcal S$ as bases, respectively.
Note that there is a one-to-one correspondence between vectors in $C_{\mathcal Q}$ and subsets of $\mathcal Q$, thus we treat them interchangeably; similarly for vectors in $C_{\mathcal S}$ and subsets of  $\mathcal S$.
Let $\partial: C_{\mathcal Q} \rightarrow C_{\mathcal S}$ be a linear map, called the boundary map, which for any Pauli $Z$ error with support $\epsilon\subseteq \mathcal Q$ returns its $X$-type syndrome $\sigma\subseteq \mathcal{S}$, i.e., $\sigma = \partial \epsilon$.
We say that a location $v\in V$ is trailing for $\sigma\in\im\partial$ iff $\sigma\rest v$ is nonempty and belongs to the future of $v$, i.e., $\sigma\rest v \subset \future v$.

Now, we proceed with defining a causal code.
We say that a quadruple $((V,\preceq), \mathcal Q, \mathcal S, \partial)$ describes a causal code iff the following conditions are satisfied.
\begin{enumerate}
\item (causal diamonds) For any finite subset of locations $U\subseteq V$ there exists the causal diamond $\dmond U$.
\item (locations) Every qubit $Q\in\mathcal Q$ and every stabilizer generator $S\in\mathcal S$ correspond to finite subsets of locations, i.e., $Q,S \subseteq V$.
\item (qubit infimum) For every qubit $Q\in\mathcal Q$ its infimum satisfies $\inf Q \in Q$.
\item (syndrome evaluation) The syndrome $\partial \epsilon$ of any error $\epsilon \subseteq \mathcal Q$ can be evaluated locally, i.e.,
\begin{equation}
\forall v\in V: (\partial \epsilon)\rest v = [\partial(\epsilon\rest v)]\rest v.
\end{equation}
\item (trailing location) For any location $v \in V$ and the syndrome $\sigma \in \im \partial$,
if $\sigma\rest v$ is nonempty and  $\sigma\rest v \subset\future v$, then there exists a subset of qubits $\varphi(v) \subseteq \mathcal Q \rest v \cap \future v$ satisfying
$[\partial \varphi (v)]\rest v = \sigma\rest v$ and $\dmond{\varphi(v)} = \dmond{\sigma\rest v}$.
\end{enumerate}

We can adapt the sweep rule to any causal code $((V,\preceq), \mathcal Q, \mathcal S, \partial)$.
We define the sweep rule for every location $v\in V$ in the same way as in Ref.~\cite{Kubica2019}.
\begin{definition}[sweep rule]
If $v$ is trailing, then find a subset of qubits $\varphi(v)\subseteq \mathcal{Q}\rest v \cap \future{v}$ with a boundary that locally matches $\sigma$, i.e.\ $[\partial\varphi(v)]\rest v=\sigma \rest v$. Return $\varphi(v)$.
\label{def:sweep}
\end{definition}

Our first result is a lemma concerning the properties of the sweep rule.
But, before we state the lemma, we must make some additional definitions.
Let $u,v\in V$ be two locations satisfying $u\prec v$.
We say that a sequence of locations $u\prec w_{1}\prec\ldots\prec w_{n}\prec v$, where $n=0,1,\ldots$, forms a chain between $u$ and $v$ of length $n+1$.
We define $\mathcal N (u,v)$ to be the collection of all the chains between $u$ and $v$.
We write $\ell (N)$ to denote the length of the chain $N\in\mathcal N (u,v)$.

Given a causal code $((V,\preceq), \mathcal Q, \mathcal S, \partial)$, we define its corresponding syndrome graph $G$ as follows.
For each stabilizer $S \in \mathcal S$, there is a node in $G$ and we add an edge between any two nodes iff their corresponding stabilizers both have a non-zero intersection with the same qubit in $\mathcal Q$.
We define the syndrome distance $d_G(S, T)$ between any two stabilizer generators $S,T\in \mathcal S$ to be the graph distance in $G$, i.e., the length of the shortest path in $G$ between the nodes corresponding to $S$ and $T$.
This can be extended to the syndromes $\sigma,\tau \subseteq \mathcal S$ in the obvious way: $d_G(\sigma, \tau) = \min_{S\in\sigma,T\in\tau}d_G(S, T)$.
 
\begin{lemma}[sweep rule properties]
Let $\sigma\in\im\partial$ be a syndrome of the causal code $((V,\preceq), \mathcal Q, \mathcal S, \partial)$.
Suppose that the sweep rule is applied simultaneously at every location in $V$ at time steps $T=1,2,\ldots$, and the syndrome is updated at each time step as follows: $\sigma^{(T+1)}=\sigma^{(T)}+\partial \varphi^{(T)}$, where $\varphi^{(T)}$ is the set of qubits returned by the rule. 
Then,
    \begin{enumerate}
        \item (support) the syndrome at time $T$, $\sigma^{(T)}$ stays within the causal diamond of the original syndrome, i.e.\
        \begin{equation}
            \sigma^{(T)}\subseteq\dmond{\sigma},
            \label{eq:lem-supp}
        \end{equation}
        \item (propagation) the syndrome distance between $\sigma$ and any $S\in \sigma^{(T)}$ is at most $T$, i.e.\
        \begin{equation}
            d_G(S,\sigma)\leq T,
            \label{eq:lem-prop}
        \end{equation}
        \item (removal) the syndrome is trivial, i.e.\ $\sigma^{(T)}=0$, for 
        \begin{equation}
            T > \max_{v \in \bigcup \sigma}\max_{N \in \mathcal N (v, \sup \sigma)}\ell(N).
        \end{equation}
        \todo{That is, the syndrome is trivial for times $T$ greater than the maximal chain length between a location $v$ and the supremum of the syndrome $\sup \sigma$, maximized over all locations $v$ contained in the syndrome (viewed as a subset of locations).}
    \end{enumerate}
    \label{lem:sr-props}
\end{lemma}
The above lemma is analogous to Lemma 2 from Ref.~\cite{Kubica2019}, but \change[it now applies to more general causal codes]{is not limited to codes defined on geometric lattices; rather, it now applies to more general causal codes.} We defer the proof until \cref{sec:methods}.

\subsection{Rhombic dodecahedral toric codes}
\label{subsec:rhombic}

In this section, we examine an example of a causal code: the 3D toric code defined on the infinite rhombic dodecahedral lattice. Toric codes defined on this lattice are of interest because they arise in the transversal implementation of $CCZ$ in 3D toric codes~\cite{Kubica2015,Vasmer2019,Browneaay4929}. Consider the tessellation of $\mathbb R^3$ by rhombic dodecahedra, where a rhombic dodecahedron is a face-transistive polyhedron with twelve rhombic faces. One can construct this lattice from the cubic lattice, as follows. Begin with a cubic lattice, where the vertices of the lattice are the elements of $\mathbb Z^3$. Create new vertices at all half-integer coordinates $(x/2,y/2,z/2)$, satisfying $x+y+z=1 \mod 4$ and $xyz = 1 \mod 2$. These new vertices sit at the centres of half of all cubes in the cubic lattice. For each such cube, add edges from the new vertex at its centre to the vertices of the cube. Finally, delete the edges of the cubic lattice. The remaining lattice is a rhombic dodecahedral lattice. \Cref{fig:rb-lat} gives an example of this procedure.

\begin{figure*}
    \centering
    \subfloat[]{
        \centering
        \includegraphics[width=0.3\textwidth]{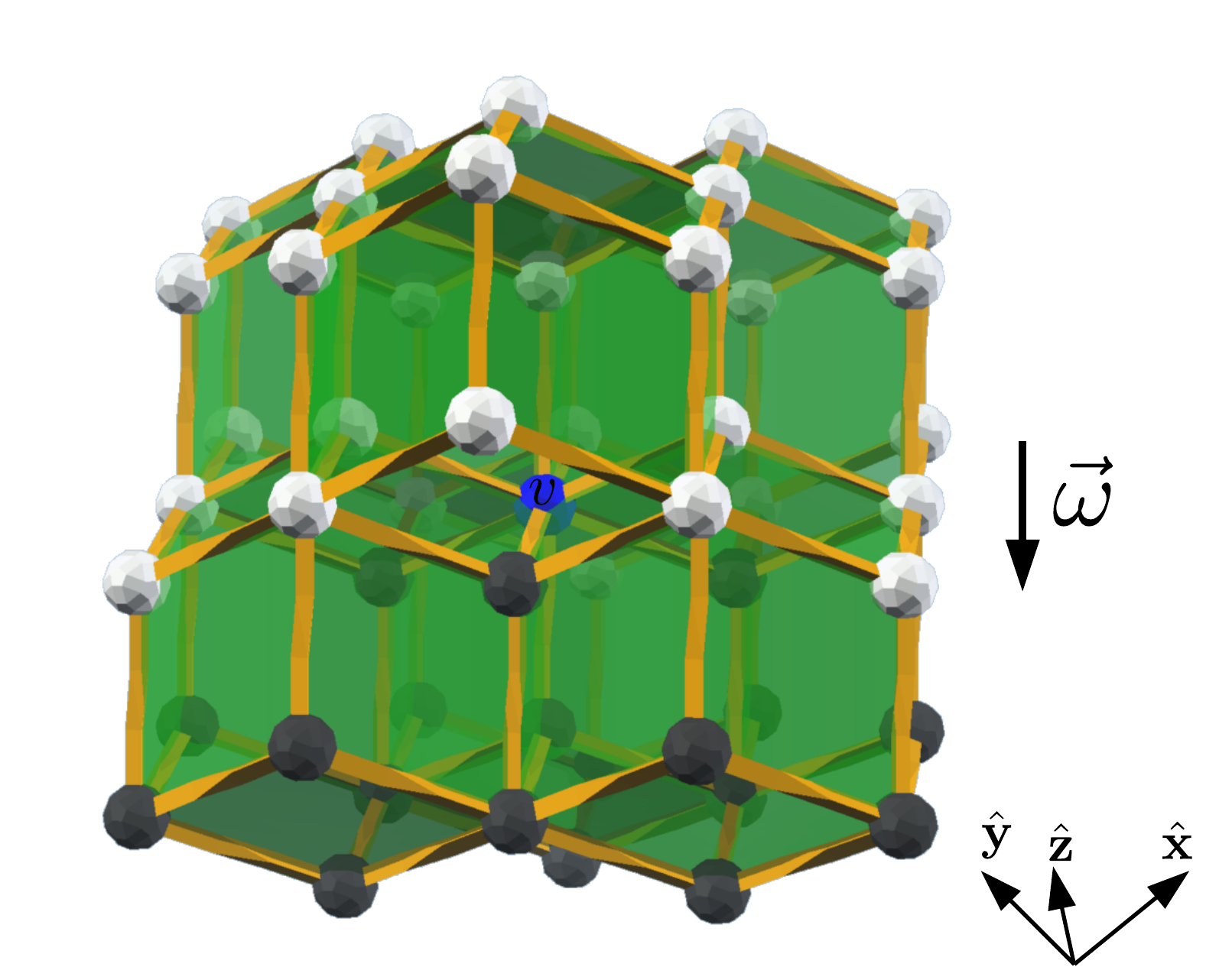}
    }
    \hfill
    \subfloat[]{
        \centering
        \includegraphics[width=0.3\textwidth]{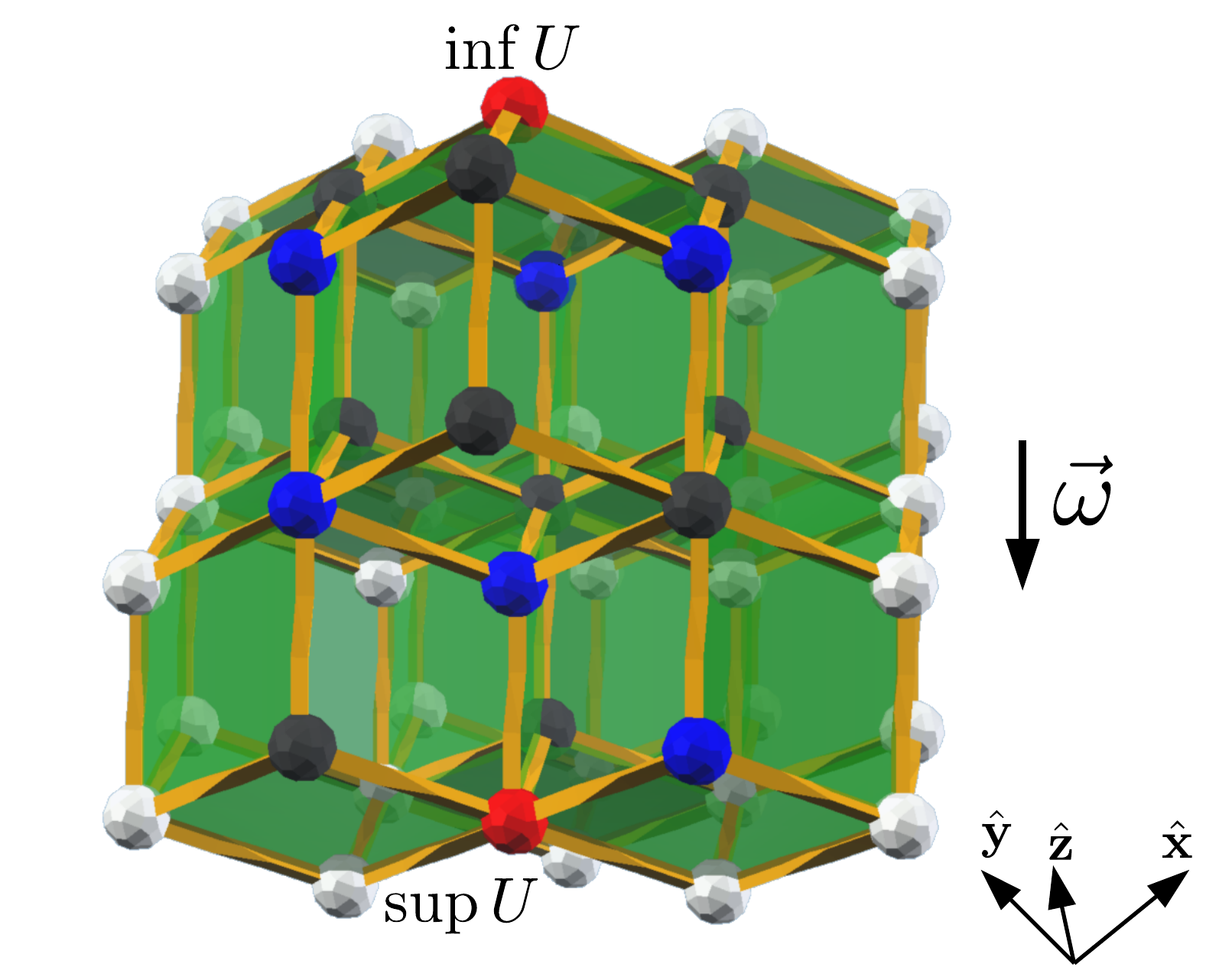}
    }
    \hfill
    \subfloat[]{
        \label{subfig:sr}
        \centering
        \includegraphics[width=0.3\textwidth]{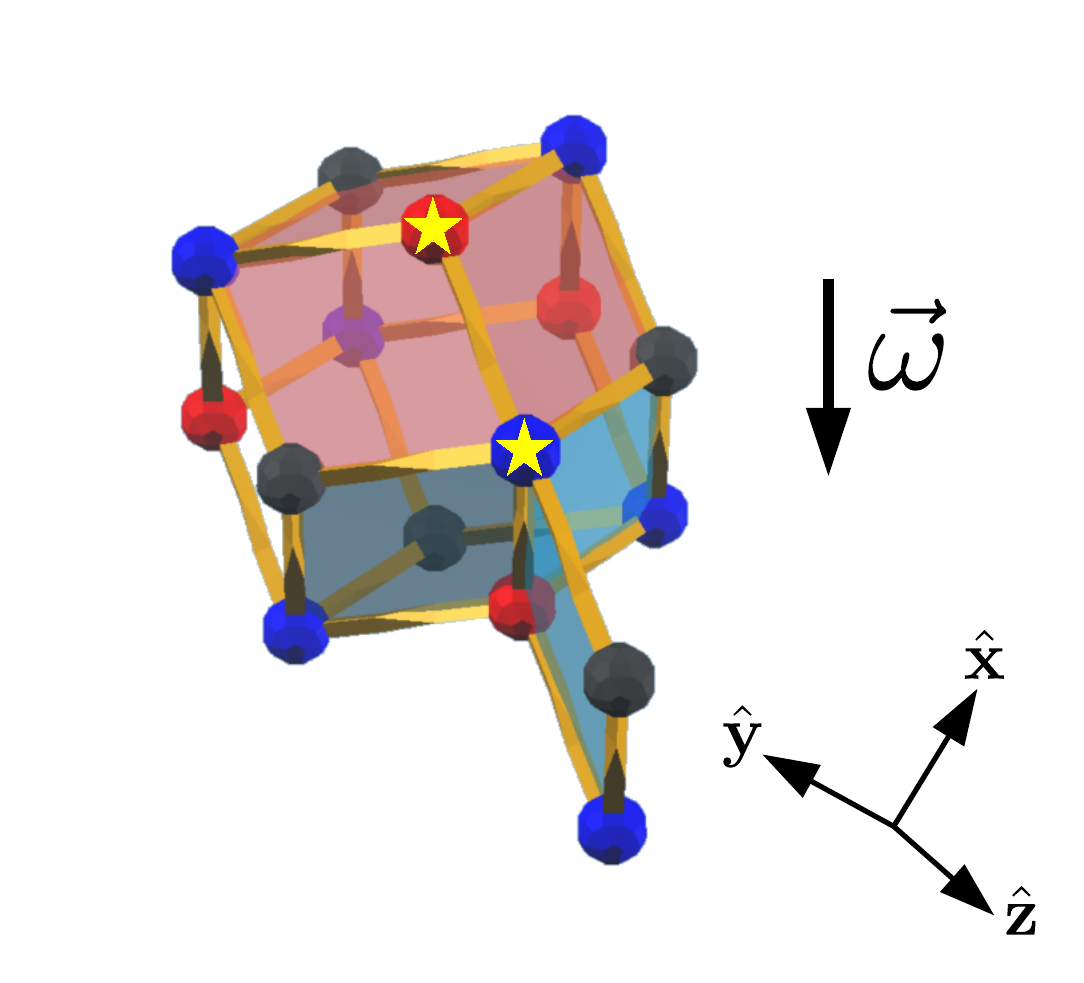}
    }
    \caption{The sweep rule in the rhombic dodecahedral lattice. The sweep direction $\vec \omega = -(1,1,1)$ is indicated by the arrow. (a) The future $\future v$ of $v$ (the blue vertex), is the power set of the black vertices and the blue vertex. (b) The causal diamond of $U$, where $U$ is the set of blue vertices. The red vertices are $\sup U$ and $\inf U$, and $\dmond U$ is the power set of the red, black and blue vertices. (c) In the rhombic dodecahedral lattice, there are two types of vertices: one type is degree four (red and black) and the other is degree eight (blue). The red (blue) shaded faces are the qubits that the rule may return, depending on the syndrome at the highlighted red (blue) vertex. The rule returns nothing at the black vertices because there are no syndromes whose restriction to a black vertex is in the future of that vertex.}
    \label{fig:sr-rd}
\end{figure*}

\todo{We denote the infinite rhombic dodecahedral lattice by $\mathcal L^\infty$, and we denote its vertices, edges, faces, and cells by $\mathcal L^\infty_0$, $\mathcal L^\infty_1$, $\mathcal L^\infty_2$, and $\mathcal L^{\infty}_3$ respectively. We place qubits on faces, and we associate $X$ and $Z$ stabilizer generators with edges and cells, respectively. That is, for each edge $e \in \mathcal L^{\infty}_1$, we have a stabilizer generator $\Pi_{f : e \in f} X_f$, and for each cell $c \in \mathcal L^{\infty}_3$, we have a stabilizer generator $\Pi_{f \in c} Z_f$, where $X_f$ ($Z_f$) denotes a Pauli $X$ ($Z$) operator acting on the qubit on face $f$.} In the notation of the previous section, $V= \mathcal L^\infty_0$, $\mathcal{S} = \mathcal L^\infty_1$, and $\mathcal Q = \mathcal L^\infty_2$. Let $\vec \omega \in \mathbb R^3$ be a vector (a sweep direction) that is not perpendicular to any of the edges of $\mathcal L^\infty$. Such a sweep direction induces a partial order over $\mathcal L^\infty_0$, as we now explain. Let $(u:v)$ denote a path from one vertex $v \in \mathcal L^\infty_0$ to another vertex $u \in \mathcal L^\infty_0$, where $(u:v)=\{ (u, w_{1}), \ldots, (w_{n}, v) \}$ is a set of edges. We call a path from $u$ to $v$ causal (denoted by $\cpath{u}{v}$) if the inner product $\vec \omega \cdot (w_{i}, w_{i+1})$ has the same sign for all edges $(w_{i}, w_{i+1}) \in \cpath{u}{v}$. We write $u \preceq v$, if $u=v$ or there exists a causal path, $\cpath{u}{v}$, and $\vec \omega \cdot (w_{i},w_{i+1}) > 0$ for all edges in the path. 

We now verify that the rhombic dodecahedral toric code equipped with a sweep direction is a causal code. \todo{We note that this is only with respect to phase-flip errors, as for bit-flip errors the trailing location condition is not satisfied. Higher-dimensional toric codes can satisfy the causal code conditions for both bit-flip and phase-flip errors, e.g.\ the 4D toric code with qubits on faces~\cite{Dennis2002,Alicki2010}.} We choose a sweep direction that is parallel to one of the edges directions of the lattice, $\vec \omega = (1,1,1)$. First, consider the causal diamond condition. One can prove by induction that any finite subset of vertices of the lattice has an infimum and supremum, and therefore has a unique causal diamond. \Cref{fig:sr-rd} shows an example of the future of a vertex and the causal diamond of a subset of vertices. By definition, the qubits and stabilizer generators are associated with faces and edges, which are finite subsets of locations (vertices), so the locations condition is satisfied. Next, consider the qubit infimum conidition. The faces of the lattice are rhombi, and as $\vec \omega$ is not perpendicular to any of the edges of the lattice, each face contains its infimum. 

Next, consider the syndrome evaluation condition. We consider vector spaces $C_1$ and $C_2$  with bases given by $e \in \mathcal L^\infty_1$ and $f \in \mathcal L^\infty_2$, respectively. This allows us to define the linear boundary operator $\partial : C_2 \rightarrow C_1$, which is specified for all basis elements $f \in C_2$ as follows: $\partial f = \sum_{e \in f} e$. In words, for each face, $\partial$ returns the sum of the edges in the face. The syndrome of a \todo{(phase-flip)} error $\epsilon \subseteq \mathcal L^\infty_2$ is then $\partial \epsilon$. This syndrome evaluation procedure is local. 

Now, we consider the final causal code condition: the trailing location condition. The rhombic dodecahedral lattice has two types of vertex: one type is degree four and the other is degree eight. Given our choice of sweep direction, one can easily verify that the trailing location condition is satisfied at all the vertices of $\mathcal L^\infty$, as illustrated in \cref{subfig:sr}.

\subsection{Sweep decoder for toric codes with boundaries}
\label{subsec:boundaries}

In this section, we consider the extension of the sweep rule to 3D toric codes defined on a family of rhombic dodecahedral lattices $\mathcal L$ with boundaries, with growing linear size $L$. We discuss the problems that occur when we try to apply the standard sweep rule to this lattice. Then, we present a solution to these problems in the form of a modified sweep decoder. 

We construct a family of rhombic dodecahedral lattices with boundaries from the infinite rhombic dodecahedral lattice, by considering finite regions of it. We start with a cubic lattice of size $(L-1)\times(L+1)\times L$, with vertices at integer coordinates $(x,y,z) \in [0, L-1] \times [0, L+1] \times [0, L]$. To construct a rhombic dodecahedral lattice, we create vertices at the centre of cubes with coordinates $(x/2,y/2,z/2)$ satisfying $x+y+z=1\mod 4$ and $xyz=1\mod 2$. Then for each vertex at the centre of a cube, we create edges from this vertex to all the vertices $(x,y,z)$ of its cube satisfying $0 < y < L+1$ and $0 < z < L$. Finally, we delete the edges of the cubic lattice and all the vertices $(x,y,z)$ of the cubic lattice with $y=0,L+1$ or $z = 0,L$. \Cref{fig:rb-lat} illustrates the construction of the rhombic dodecahedral lattice $\mathcal L$ for $L=3$. 

\begin{figure*}
    \centering
    \subfloat[]{
        \centering
        \includegraphics[width=0.27\textwidth]{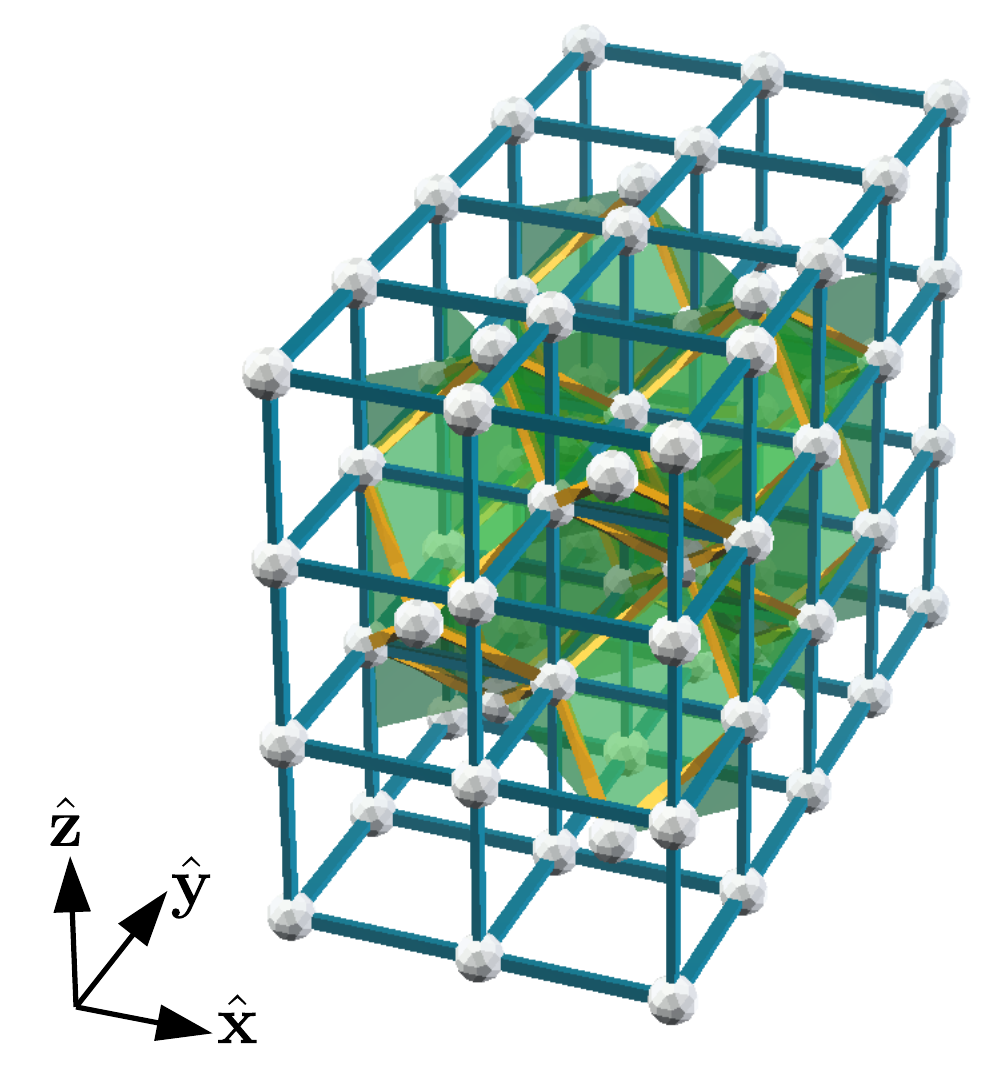}
    }
    \hfill
    \subfloat[]{
        \label{subfig:rs}
        \centering
        \includegraphics[width=0.22\textwidth]{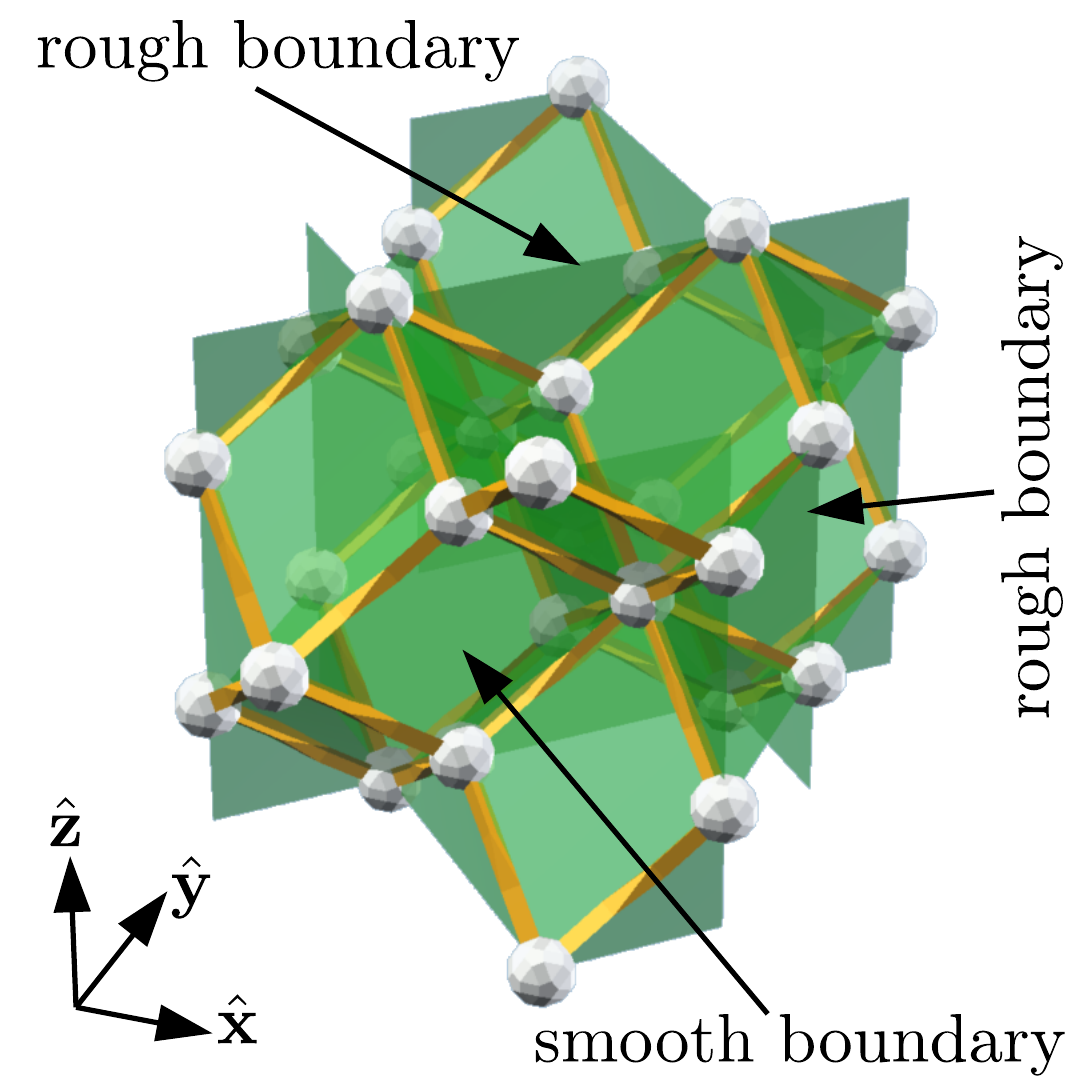}
    }
    \hfill
    \subfloat[]{
        \centering
        \includegraphics[width=0.22\textwidth]{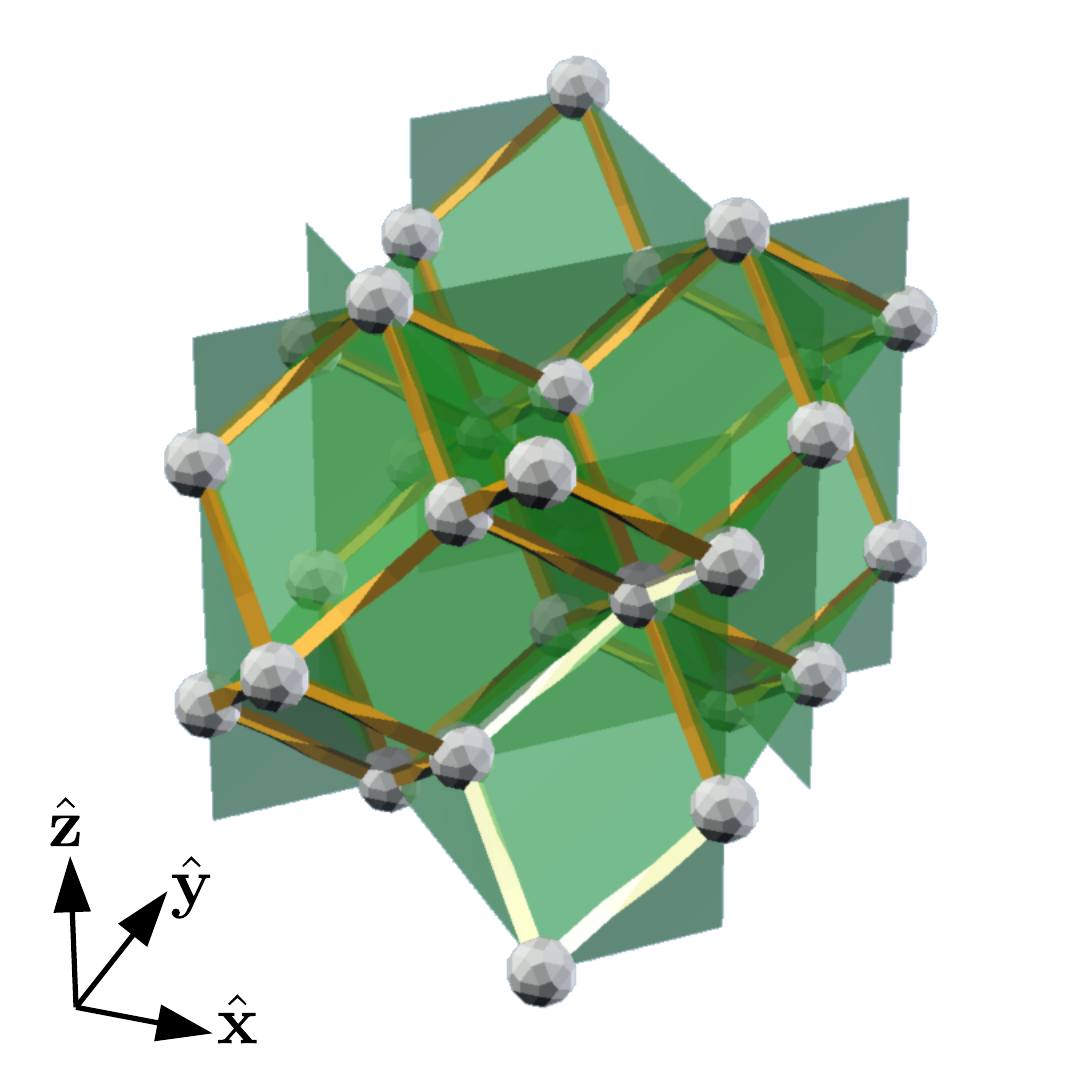}
    }
    \hfill
    \subfloat[]{
        \label{subfig:smooth}
        \centering
        \includegraphics[width=0.22\textwidth]{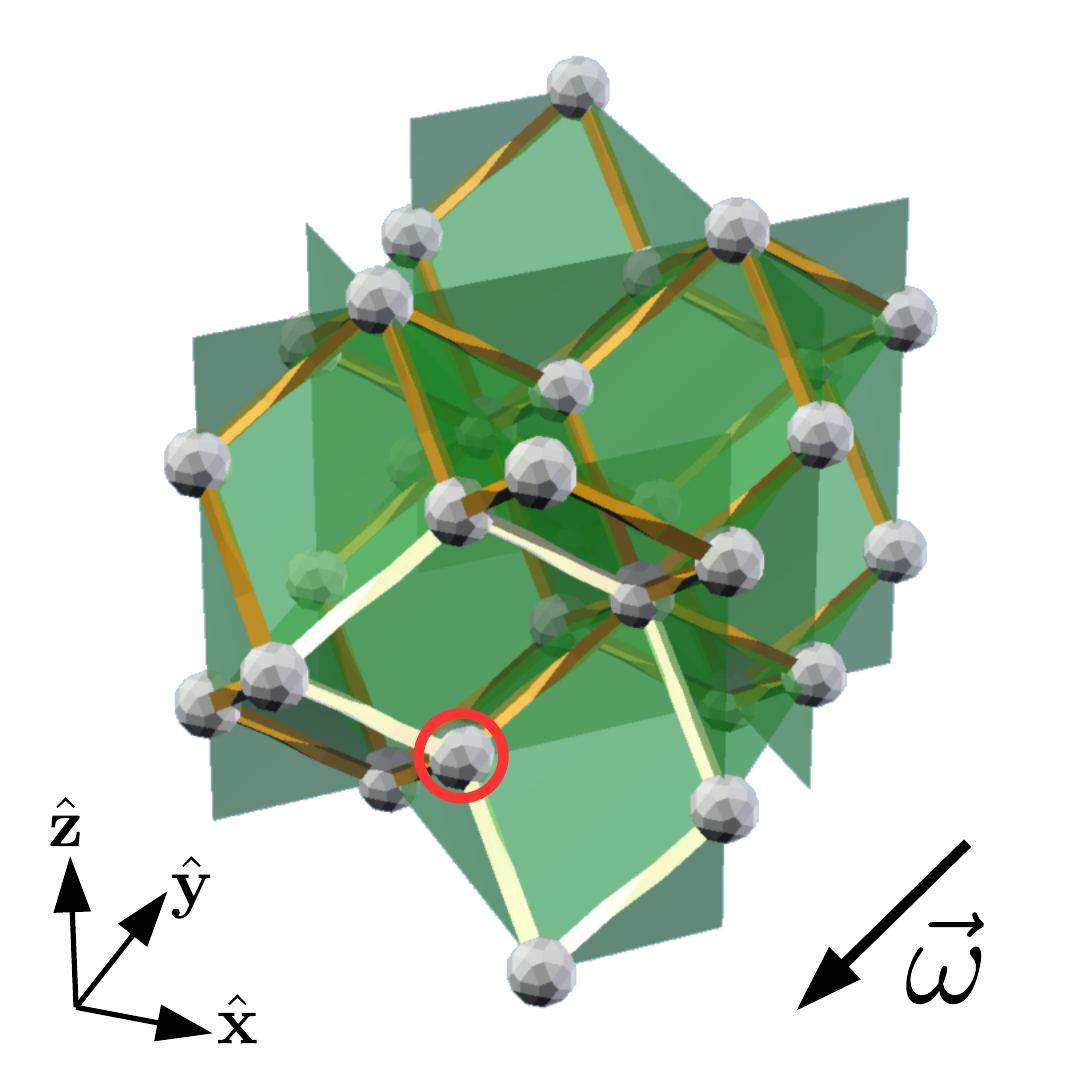}
    }
    \caption{A family of rhombic dodecahedral lattices with boundaries. (a) Construction of the $L=3$ lattice from the cubic lattice. We show the initial cubic lattice as well as the final rhombic dodecahedral lattice. (b) The $L=3$ lattice. The front and back boundaries are smooth (syndromes must form closed loops) and the other boundaries are rough (open loops of syndrome can terminate). (c) An example of a open loop of syndrome terminating on one of the rough boundaries (light yellow edges). (d) A syndrome on the smooth boundary (light yellow edges). With the sweep direction $\vec \omega = -(1,1,1)$, the ringed vertex does not satisfy the trailing location condition.}
    \label{fig:rb-lat}
\end{figure*}

We denote the vertices, edges, and faces of the lattice $\mathcal L$ by $\mathcal L_0$, $\mathcal L_1$, and $\mathcal L_2$, respectively. As in the previous section, we associate qubits with the faces of the lattice and we associate $X$ stabilizer generators with the edges. We define a local region of $\mathcal L$ to be a region of $\mathcal L$ with diameter smaller than $L/2$. In particular, there are no non-trivial logical $Z$ operators supported within a local region. The boundary map $\partial$ is defined in the same way as the infinite case, except for an important caveat. Some faces in $\mathcal L$ only have one or two edges (see \cref{fig:rb-lat}), so $\partial$ can return the sum of fewer than four edges. In the bulk of the lattice, syndromes must form closed loops. As illustrated in \cref{subfig:rs}, the lattices in our family have two types of boundaries: rough and smooth. Open loops of syndrome may only terminate on the rough boundaries.

For a chosen sweep direction, some vertices on the boundaries of $\mathcal L$ will not satisfy the trailing location condition (see \cref{subfig:smooth} for an example). This means that some syndromes near the boundary are immobile under the action of the Sweep rule. Clearly, this poses a significant problem for a decoder based on the Sweep rule, as there are some constant weight errors whose syndromes are immobile. Fortunately, there is a simple solution to this problem: periodically varying the sweep direction. We pick a set of eight sweep directions 
\begin{equation}
    \Omega = \{ (\pm 1, \pm 1, \pm 1) \},    
\end{equation}
where each $\vec \omega \in \Omega$ is parallel to one of the eight edge directions of $\mathcal L$. The rule for each direction is analogous to that shown in \cref{subfig:sr}. 

For our set of sweep directions, there are no errors within a local region whose syndromes are immobile for every sweep direction. To make this statement precise, we need to adapt the definition of a causal diamond to lattices with boundaries. As $\mathcal L$ is a subset of the infinite rhombic dodecahedral lattice $\mathcal L^\infty$, the causal diamond of $U \subseteq \mathcal L_0$ is well defined in $\mathcal L^\infty$ for any $\vec \omega \in \Omega$, but it may contain subsets of vertices that are not in $\mathcal L$. Roughly speaking, we define the causal region of $U$ with respect to $\vec \omega$, $\mathcal R_{\vec \omega} (U)$, to be the restriction of the causal diamond in $\mathcal L^\infty$ (with respect to $\vec \omega$) to the finite lattice $\mathcal L$. More precisely, it is a subset of the elements of the causal diamond, whose vertices belong to $\mathcal L$, i.e.\ $\mathcal R_{\vec \omega} (U) = 2 ^ {\lozenge_{\vec \omega} (U) \cap \mathcal L_0}$, where $\lozenge_{\vec \omega} (U)$ is the causal diamond with respect to $\vec \omega$. We now state a lemma that is sufficient to prove that a decoder based on the Sweep rule has a non-zero error threshold for rhombic dodecahedral lattices with boundaries. 

\begin{lemma}
    % Let $\Omega$ be a set of sweep directions and $\sigma \in \im \partial$ be a syndrome, such that $\mathcal R_{\vec \omega} (\sigma)$ is contained in a local region of $\mathcal L$ for all $\vec \omega \in \Omega$. Then there exists a sweep direction $\vec \omega^* \in \Omega$ such that the trailing location condition is satisfied at every vertex in $\mathcal R_{\vec \omega^*} (\sigma)$.
    \todo{Let $\Omega$ be a set of sweep directions and $\sigma \in \im \partial$ be a syndrome, such that, for all sweep directions $\vec \omega \in \Omega$, the causal region of $\sigma$ with respect to $\vec \omega$,  $\mathcal R_{\vec \omega} (\sigma)$, is contained in a local region of $\mathcal L$.}
    Then there exists a sweep direction $\vec \omega^* \in \Omega$ such that the trailing location condition is satisfied at every vertex in $\mathcal R_{\vec \omega^*} (\sigma)$.
    \label{lem:trailing}
\end{lemma}

We defer the proof of \cref{lem:trailing} until \cref{sec:methods}. We now give a pseudocode description of a modified sweep decoder that works for 3D toric codes defined on lattices with boundaries. \todo{We note that all addition below is carried out modulo 2, as we view errors and syndromes as $\mathbb{F}_2$ vectors (see \cref{subsec:rhombic}).}

\begin{algorithm}[H]
    \caption{\textsc{Sweep Decoder}}
    \label{alg:srd}
    \begin{algorithmic}
        \Require Error syndrome $\sigma \in \im \partial$, set of sweep directions $\Omega$, runtime per direction $T_{max}$
        \Ensure Correction $\varphi \subseteq \mathcal L_2$
        \State $\varphi \gets 0$
            \For{$j\gets 1$ \textbf{to} $|\Omega|$}
                \State $T\gets 1$
                \State $\vec \omega \gets \Omega[j]$
                \algorithmiccomment{\todo{Change sweep direction.}}
                \While{$T \leq T_{max}$}
                     \State \multiline{%
                     Apply the sweep rule with sweep direction $\vec \omega$ simultaneously to every vertex of $\mathcal{L}$ to get $\varphi^{(T)} \subseteq \mathcal L_2$, the set of qubits returned by the rule.}
                    \State $\sigma \gets \sigma + \partial \varphi^{(T)}$ \algorithmiccomment{\todo{Update syndrome.}}
                    \State $\varphi \gets \varphi + \varphi^{(T)}$ 
                    \algorithmiccomment{\todo{Update correction.}}
                    \State $T\leftarrow T+1$
                \EndWhile
            \EndFor
        \If{$\sigma \neq 0$}
            \State \Return \textsc{fail}
        \Else
            \State \Return $\varphi$
        \EndIf
\end{algorithmic}
\end{algorithm}

The decoder can fail in two ways. Firstly, if the syndrome is still non-trivial after $|\Omega| \times T_{max}$ applications of the sweep rule, we consider the decoder to have failed. Secondly, the decoder can fail because the product of the correction and the original error implements a non-trivial logical operator. We can now state our main theorem. 

\begin{theorem}
    Consider 3D toric codes defined on a family of rhombic dodecahedral lattices $\mathcal L$, with growing linear size $L$. There exists a constant $p_{\mathrm{th}}>0$ such that for any phase-flip error rate $p<p_{\mathrm{th}}$, the probability that the sweep decoder fails is $\BigO{(p/p_{\mathrm{th}})^{\beta_1 L^{\beta_2}}}$, for some constants $\beta_1,\beta_2 > 0$.
    \label{thm:threshold}
\end{theorem}

We provide a sketch of the proof here, and postpone the details until Supplementary Note 1. Our proof builds on~\cite{Kubica2019,Bravyi2013,Harrington2004,Gacs1988} and relies on standard results from the literature~\cite{Bravyi2013,Van1985}.

\begin{proof}
    First, we consider a chunk decomposition of the error. This is a recursive definition, where the diameter of the chunk is exponential in the recursion level. Next, we use the properties of the sweep rule (\cref{lem:sr-props}) to show that the sweep decoder successfully corrects chunks up to some level $m^{*}=\BigO{\log L}$. To accomplish this, we first show that the sweep decoder is guaranteed to correct errors whose diameter is smaller than $L$, in a number of time steps that scales linearly with the diameter. Secondly, we rely on a standard lemma that states that a connected component of a level-$m$ chunk is well separated from level-$n$ chunks with $n\geq m$, which means that the sweep decoder corrects connected components of the error independently. Finally, percolation theory tells us that the probability of an error containing a level-$n$ chunk is $\BigO{(p/p_{\mathrm{th}})^{2^n}}$, for some $p_{\mathrm{th}}>0$. As the decoder successfully corrects all level-$n$ chunks for $n<m^{*}=\BigO{\log L}$, the failure probability of the decoder is $\BigO{(p/p_{\mathrm{th}})^{\beta_1 L^{\beta_2}}}$.
\end{proof}

\subsection{Numerical implementation and optimization \label{subsec:num-meas}}

We implemented the sweep decoder in C++ and simulated its performance for 3D toric codes defined on rhombic dodecahedral lattices, with and without boundaries. The code is available online \footnote{See the repository at \url{https://github.com/MikeVasmer/Sweep-Decoder-Boundaries} for a C++ implementation of the sweep decoder, as well as interactive models of the lattices with boundaries.}. We study the performance of the sweep decoder for error models with phase-flip and measurement errors. Specifically, we simulate the following procedure.

\begin{definition}[Decoding with noisy measurements]
    Consider the 3D toric code with qubits on faces and $X$ stabilizers on edges. At each time step $T\in\{1,\ldots,N\}$, the following events take place:
    \begin{enumerate}
        \item A $Z$ error independently affects each qubit with probability $p$.
        \item The $X$ stabilizer generators are measured perfectly.
        \item Each syndrome bit is flipped independently with probability $q=\alpha p$, where $\alpha$ is a free parameter.
        \item The sweep rule is applied simultaneously to every vertex, and $Z$ corrections are applied to the qubits on faces returned by the rule.
    \end{enumerate}
    After $N$ time steps have elapsed, the $X$ stabilizer generators are measured perfectly and we apply \cref{alg:srd}. Decoding succeeds if, and only if, the product of the errors and corrections (including the correction returned by \cref{alg:srd}) is a $Z$-type stabilizer.
    \label{def:problem}
\end{definition}

We note that measuring the stabilizers perfectly after $N$ time steps may seem unrealistic. However, this can model readout of a CSS code, because measurement errors during destructive single-qubit $X$ measurements of the physical qubits have the same effect as phase-flip errors immediately prior to the measurements. 

We first present our results for lattices with periodic boundary conditions i.e.\ each lattice is topologically a 3-torus. Although changing the sweep direction is not necessary for such lattices, we observe improved performance compared with keeping a constant sweep direction. We consider error models with phase-flip error probability $p$ and measurement error probability $q = \alpha p$, for various values of $\alpha$. For a given error model, we study the error threshold of the decoder as a function of $N$, the number of cycles, where a cycle is one round of the procedure described in \cref{def:problem}. For a range of values of $N$, we estimate the logical error rate, $p_{\mathrm L}$, as a function of $p$ for different values of $L$ (the linear lattice size, or equivalently the code distance of the toric code). We estimate the error threshold as the value of $p$ at which the curves for different $L$ intersect. We find that the error threshold decays polynomially in $N$ to a constant value, the sustainable threshold. The sustainable threshold is a measure of how much noise the decoder can tolerate over a number of cycles much greater than the code distance. We use the following numerical ansatz 
\begin{equation}
    p_{\mathrm{th}}(N) = p_{\mathrm{sus}} \left[1 - \left(1 - \frac{p_{\mathrm{th}}(1)}{p_{\mathrm{sus}}}\right)N^{-\gamma}\right],
    \label{eq:sus-thresh}
\end{equation}
to model the behaviour of the error threshold as a function of $N$, where $\gamma$ and $p_{\mathrm{sus}}$ (the sustainable threshold) are parameters of the fit. For an error model with $\alpha = q/p = 1$, we find a sustainable threshold of $p_{\mathrm{sus}} \approx 2.1\%$, with $\gamma = 1.06$ and $p_{\mathrm{th}}(1) = 21.5\%$ (as shown in \cref{fig:rhombic-toric}).

We observe that the decoder has a significantly higher tolerance to measurement noise as opposed to qubit noise, as shown in \cref{fig:rhombic-toric}. We find that the maximum phase-flip error rate that the decoder can tolerate is $p \approx 2.9\%$ (for $q=0$), compared with a maximum measurement error rate of $q \approx 8\%$ (for $p \rightarrow 0^+$). Our results show that the sweep decoder has an inbuilt resilience to measurement errors. To understand why this is the case, let us analyse the effect of measurement errors on the decoder. First, consider measurement errors that are far from phase-flip errors. A single isolated measurement error cannot cause the decoder to erroneously apply a Pauli-$X$ operator. To deceive the decoder, two measurement errors must occur next to each other, such that a vertex becomes trailing. Second, measurement errors in the neighbourhood of phase-flip errors can interfere with the decoder, and prevent it from applying a correction. Thus, to affect the performance of the decoder, a single measurement error either has to occur close to a phase-flip measurement error. This explains why the sweep decoder has a higher tolerance of measurement errors relative to phase-flip errors.

We also simulated the performance of the decoder for lattices with boundaries. We consider toric codes defined the family of rhombic dodecahedral lattices with boundaries that we described in \cref{subsec:boundaries}. We find that the sustainable threshold of toric codes defined on this lattice family is $p_{\mathrm{sus}} \approx 2.1\%$, for an error model where $\alpha = q/p =  1$. This value matches the sustainable threshold for the corresponding lattice with periodic boundary conditions, as expected.

A natural question to ask is whether applying the sweep rule multiple times per stabilizer measurement improves the performance of the decoder. Multiple applications of the rule could be feasible if gates are much faster than measurements, as in e.g.\ in trapped-ion qubits~\cite{Crain2019}. We found that increasing the number of applications of the rule per syndrome measurement significantly improved the performance of the decoder for a variety of error models (including error models where $\alpha = q/p > 1$); see  \cref{fig:rate} for an example. 

\begin{figure}[!ht]
    \centering
    \includegraphics[width=0.8\textwidth]{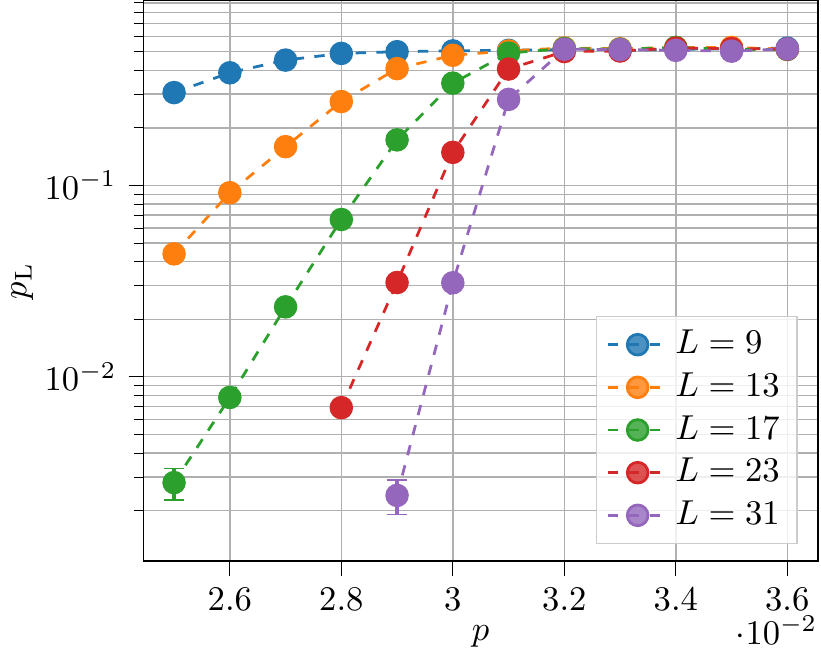}
    \caption{Improving the performance of the sweep decoder applied to the toric code on the rhombic dodecahedral lattice (with boundaries) by applying the rule multiple times per syndrome measurement. We set $\alpha = q / p = 1$ and we fix $N=2^{10}$ error correction cycles. We plot the logical error rate $p_{\mathrm L}$ as a function of $p$ for different linear lattice sizes $L$. We applied the rule three times per syndrome measurement and we observe an error threshold of $p_{\mathrm{th}} \approx 3.2\%$, an improvement of over the corresponding error threshold of $p_{\mathrm{th}} \approx 2.17\%$ when we applied the rule once per syndrome measurement (see the inset of \cref{fig:rhombic-toric}a). We use $10^4$ Monte Carlo samples for each point.}
    \label{fig:rate}
\end{figure}

The ability to change the sweep direction gives us parameters that we can use to tune the performance of the decoder. They are: the frequency with which we change the sweep direction and the order in which we change the sweep direction. We investigated the effect of varying both of these parameters. The most significant parameter is the direction-change frequency. Our simulation naturally divides into two phases: the error suppression phase where the rule applied while errors are happening, and the perfect decoding phase where the rule is applied without errors. In the error suppression phase, we want to prevent the build-up of errors near the boundaries so we anticipate that we may want to vary the sweep direction more frequently. We find that changing direction after $\sim \log L$ sweeps in the error-suppression phase and $L$ sweeps in the perfect decoding phase gave the best performance, as shown in \cref{fig:freq}. We find that the order in which we change the sweep direction does not appreciably impact the performance of the decoder. In addition, we find that the performance of the regular sweep decoder was superior to the greedy sweep decoder introduced in~\cite{Kubica2019}. 

\begin{figure}
    \centering
    \includegraphics[width=0.8\textwidth]{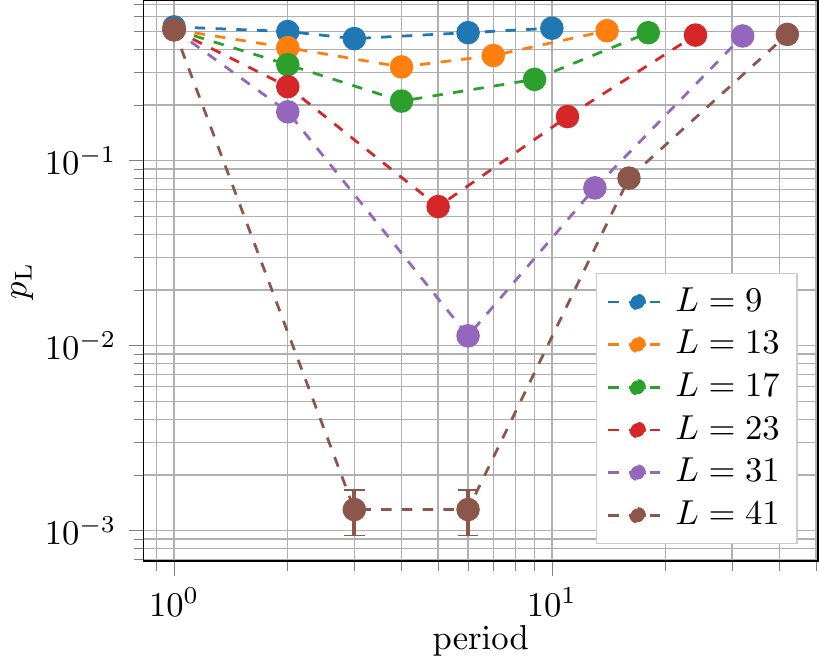}
    \caption{Optimizing the direction-change frequency for the sweep decoder applied to the toric code on the rhombic dodecahedral lattice (with boundaries). We plot the logical error rate $p_{\mathrm L}$ as a function of the direction-change period for various values of $L$. The number of error correction cycles is $N=2^{10}$ and $p=q=0.021$. We achieve the best performance when we change sweep direction every $\sim \log L$ cycles. We use $10^4$ samples for each point.}
    \label{fig:freq}
\end{figure}

Finally, we evaluated the performance of the sweep decoder against a simple correlated noise model, finding a reduced error threshold (see Supplementary Note 3). 

\section{Discussion \label{sec:discussion}}

In this article, we extended the definition of the sweep rule CA to causal codes. We also proved that the sweep decoder has a non-zero error threshold for 3D toric codes defined on a family of lattices with boundaries. In addition, we benchmarked and optimized the decoder for various 3D toric codes. 

We now comment on the performance of the sweep decoder compared with other decoding algorithms. Recall that in \cref{tab:pth_compare}, we list the error thresholds obtained from numerical simulations for various decoders applied to toric codes defined on the cubic lattice \todo{subject to phase-flip and measurement noise} (see Supplementary Note 2 for details on the sweep decoder numerics). Although the sweep decoder does not have the highest error threshold, it has other advantages that make it attractive. Firstly, as it is a CA decoder, it is highly parallelizable, which is an advantage when compared to decoding algorithms that require more involved processing such as the RG decoder introduced in Ref.~\cite{Duivenvoorden2019}. In addition, whilst neural network decoders such as~\cite{Breuckmann2018} have low complexity once the network is trained, for codes with boundaries the training cost may scale exponentially with the code distance~\cite{Maskara2019}. Also, the sweep decoder is the only decoder in \cref{tab:pth_compare} that exhibits single-shot error correction. In contrast, using the RG decoder it is necessary to repeat the stabilizer measurements $\BigO{L}$ times before finding a correction, which further complicates the decoding procedure. Finally, we note there is still a large gap between the highest error thresholds in \cref{tab:pth_compare} and the theoretical maximum error threshold of $p_{\mathrm{th}} = 23.180(4)\%$, as predicted by statistical mechanical mappings~\cite{Hasenbusch2007,Ozeki1998, Ohno2004, Kubica2018}. 

The sweep decoder could also be used in other topological codes with boundaries. Recently, it was shown that toric code decoders can be used to decode color codes~\cite{Kubica2019a}. Therefore, we could use the sweep decoder to correct errors in $(d \geq 3)$-dimensional color codes with boundaries, as long as the toric codes on the restricted lattices of the color codes are causal codes. 

We expect that the sweep decoder would be well-suited to (just-in-time) JIT decoding, which arises when one considers realizing a 3D topological code using a constant-thickness slice of active qubits. In this context, the syndrome information from the full 3D code is not available, but we still need to find a correction. The sweep decoder is ideally suited to this task, because it only requires local information to find a correction. Let us consider the JIT decoding problem described in~\cite{Browneaay4929} as an example. We start with a single layer of 2D toric code. Then we create a new 2D toric code layer above the first, with each qubit prepared in the $\ket{0}$ state. Next, we measure the $X$ stabilizers of the 3D toric code that link the two 2D layers. Then, we measure the qubits of the lower layer in the $Z$ basis. We iterate this procedure until we have realized the full 3D code. 

The $X$ stabilizer measurements will project two layers of 2D toric code into a single 3D toric code, up to a random Pauli $Z$ error applied to the qubits. This error will have a syndrome consisting of closed loops of edges. The standard decoding strategy is simply to apply an $Z$ operator whose boundary is equal to the closed loops. However, when measurement errors occur, the syndrome will consist of broken loops. We call the ends of broken loops breakpoints. To decode with measurement errors, one can first close the broken loops (pair breakpoints) using the minimum-weight perfect matching algorithm, before applying an $Z$ correction. However, in the JIT scenario, some breakpoints may need to be paired with other breakpoints that will appear in the future. In~\cite{Browneaay4929}, Brown proposed deferring the pairing of breakpoints until later in the procedure, to reduce the probability of making mistakes. However, given the innate robustness of the sweep decoder, instead of using the `repair-then-correct' decoder outlined above, we could simply apply the sweep decoder at every step and not worry about the measurement errors. Given the results presented in \cref{subsec:num-meas}, we anticipate that the sweep decoder would be effective in this case, and so provides an alternative method of JIT decoding to that is worth exploring. 

We finish this section by suggesting further applications of the sweep decoder, such as decoding more general quantum LDPC codes, e.g.\ homological product codes~\cite{Tillich2014}, especially those introduced in~\cite{Zeng2019}. The abstract reformulation of the sweep rule CA presented in \cref{sec:intro} provides a clear starting point for this task. \todo{However, we emphasize that it is still uncertain as to whether the sweep decoder would work in this case, as more general LDPC codes may not share the properties of topological codes that are needed in the proof of \cref{thm:threshold}.} In addition, we would like to prove the existence of a non-zero error threshold when measurements are unreliable. From our investigation of this question, it seems that the proof technique in Supplementary Note 1 breaks down for this case.  

\section{Methods \label{sec:methods}}

In this section, we prove \cref{lem:sr-props}, which concerned the properties of the sweep rule in an abstract setting. In addition, we show that the sweep rule retains essentially the same properties for rhombic dodecahedral lattices with boundaries. We require these properties to prove a non-zero error threshold for the sweep decoder (see Supplementary Note 1). We begin with a useful lemma about causal diamonds, which we will use throughout this section.

\begin{lemma}
    % Let $V$ be a partially ordered set such that for any finite subset $U,W\subseteq V$, there exists the causal diamond $\dmond U$. If $U\subseteq W$, then 
    \todo{Let $V$ be a partially ordered set. For any finite subsets $U,W\subseteq V$, if $U\subseteq W$, then}
    \begin{equation}
        U\subseteq\dmond{U}\subseteq\dmond{W}.
    \end{equation}
    \label{lem:c-diamonds}
\end{lemma}

\begin{proof}
    We recall that $\dmond{U}= \future{\inf U} \cap \past{\sup U}$. As both $\future{\inf U}$ and $\past{\sup U}$ contain $U$, their intersection also contains $U$, i.e.\ $U\subseteq\dmond{U}$. If $U\subseteq W$, then $\inf W\preceq\inf U$ so $\future{\inf U}\subseteq\future{\inf W}$. Likewise, if $U\subseteq W$, then $\sup U\preceq\sup W$ and therefore $\past{\sup U}\subseteq\past{\sup W}$. Consequently, $\dmond{U}\subseteq\dmond{W}$.
\end{proof}

\subsection{Proof of sweep rule properties \label{subsec:prop-proof}}

\begin{proof}[Proof of \cref{lem:sr-props}]
    First, we prove the support property by induction. At time step $T=1$ (before the rule is applied), this property holds. Now, consider the syndrome at time $T$, $\sigma^{(T)}$. Let $U^{(T)}$ denote the set of trailing locations of the syndrome at time step $T$. Between time steps $T$ and $T+1$, for each trailing location $u\in U^{(T)}$, the sweep rule will return a subset of qubits $\varphi^{(T)}(u)$ with the property that $[\partial \varphi^{(T)}(u)]\rest u = \sigma^{(T)}\rest u$. Therefore, the syndrome at time step $T+1$ is
    \begin{equation}
        \sigma^{(T+1)} = \sigma^{(T)} + \sum_{u\in U^{(T)}} \partial \varphi^{(T)}(u).
        \label{eq:supp1}
    \end{equation}
    By assumption, $\dmond{\partial\varphi (u)}=\dmond{\sigma^{(T)}\rest u}$, and $\dmond \sigma^{(T)} \subseteq \dmond \sigma$. Making multiple uses of \cref{lem:c-diamonds}, we have
    \begin{equation}
    \begin{split}
        \sigma^{(T+1)} &\subseteq \dmond{\sigma^{(T+1)}} = \dmond{\sigma^{(T)} \cup \bigcup_{u\in U^{(T)}}\partial \varphi^{(T)}(u)}, \\
        &\subseteq \dmond{\dmond{\sigma^{(T)}} \cup \bigcup_{u\in U^{(T)}}\dmond{\partial \varphi^{(T)}(u)}} \subseteq \dmond \sigma.
    \end{split}
    \label{eq:supp2}
    \end{equation}
    
    Next, we prove the propagation property, also by induction. The property is true at time step $T=1$. Now, we prove the inductive step from time step $T-1$ to $T$. As long as $\sigma^{(T)}\neq 0$, for every $S\in\sigma^{(T)}$, either $S\in\sigma^{(T-1)}$ or there exists an edge in the syndrome graph between the node corresponding to $S$ and a node corresponding to $S'\in\sigma^{(T-1)}$. By invoking the triangle inequality, we conclude that
    \begin{equation}
        d_G(S,\sigma)\leq d_G(S, S') + d_G(S', \sigma) = 1 + (T-1) = T.
        \label{eq:prop1}
    \end{equation}
    
    To prove the removal property, we define a function
    \begin{equation}
        f_{\sigma}(T) = \max_{v \in \bigcup \sigma^{(T)}}\max_{N \in \mathcal N (v, \sup \sigma)} \ell(N).
        \label{eq:rm1}
    \end{equation}
    In words, $f_{\sigma}(T)$ is the length of the longest chain between any location $v\in \bigcup \sigma^{(T)}$ and the supremum of the original syndrome, $\sup \sigma$. If $\sigma^{(T)} = \emptyset$, then we set $f_{\sigma}(T) = 0$. We now show that $f_{\sigma}(T)$ is a monotonically decreasing function of $T$. At time $T$, any location $v\in \bigcup \sigma^{(T)}$ that maximizes the value of $f(T)$ will necessarily be trailing. Between time steps $T$ and $T+1$, a subset of qubits $\varphi^{(T)}(v) \in \future v$ will be returned, and the syndrome will be modified such that $v\notin \bigcup \sigma^{(T+1)}$. Instead, there will be new locations $u\in\bigcup \sigma^{(T+1)}$, where every $u \succ v$, which implies that $f_{\sigma}(T+1) < f_{\sigma}(T)$. We note that because every qubit $Q \in \mathcal Q$ contains its unique infimum, it is impossible for a qubit to be returned multiple times by the rule in one time step. 
    
    The removal property immediately follows from the monotonicity of $f_{\sigma}(T)$. As we consider lattices with a finite number of locations, $f_{\sigma}(1)=\max_{v \in \bigcup \sigma}\max_{N \in \mathcal N (v, \sup \sigma)} \ell(N)$ will be finite. And between each time step, $f_{\sigma}(T)$ decreases by at least one, which implies that $\sigma^{(T)}=0$ for all $T > \max_{v \in \bigcup \sigma}\max_{N \in \mathcal N (v, \sup \sigma)}\ell(N)$. 
\end{proof}

\subsection{Sweep rule properties for rhombic dodecahedral lattices with boundaries}
\label{sec:proof-bound-sketch}

In this section, we show that the sweep rule retains the support, propagation and removal properties for rhombic dodecahedral lattices with boundaries, with some minor modifications. 

We recall that we use a set of eight sweep directions $\Omega = (\pm 1, \pm 1, \pm 1)$. By inspecting \cref{fig:sr-rd}, one can verify that no $\vec \omega \in \Omega$ is perpendicular to any of the edges of the rhombic dodecahedral lattice, so the partial order is always well defined. In \cref{subsec:boundaries}, we neglected a subtlety concerning causal regions. Consider the causal region of $U \subseteq \mathcal L_0$ with respect to $\vec \omega$, $\mathcal R_{\vec \omega} (U) = 2 ^ {\lozenge_{\vec \omega} (U) \cap \mathcal L_0}$. For a given $U$, causal regions with respect to different sweep directions may not be the same. Therefore, we must modify the definition of the causal region. Let $\{\vec \omega_1, \vec \omega_2, \ldots, \vec \omega_8 \}$ be an ordering of the sweep directions $\vec \omega_j \in \Omega$. We recursively define the causal region of $U$ to be 
\begin{equation}
    \reg U = \mathcal R_{\vec \omega_8} \circ \ldots \mathcal \circ \mathcal R_{\vec \omega_2} \circ \mathcal R_{\vec \omega_1}(U),
\end{equation}
i.e.\ to compute $\reg U$, we take the causal region of $U$ with respect to $\vec \omega_1$, then we take the causal region of $\mathcal R_{\vec \omega_1}(U)$ with respect to $\vec \omega_2$, and similarly until we reach $\vec \omega_8$. 

The first step in showing that the sweep rule has the desired properties is to prove \cref{lem:trailing}. This lemma is sufficient for proving the removal property of the rule. 

\begin{proof}[Proof of \cref{lem:trailing}]
    Any vertex not on the boundaries of $\mathcal L$ satisfies the trailing location condition for all $\vec \omega \in \Omega$, so we only need to check the vertices on the boundaries. 
    
    First, we consider the rough boundaries. On each such boundary, there are vertices that do not satisfy the trailing vertex condition for certain sweep directions. Let us examine each rough boundary in turn. First, consider the vertices on the top rough boundary, an example of which is highlighted in \cref{subfig:roughtop}. For these vertices, the problematic sweep directions are those that point inwards, i.e.\ $\vec \omega = (x,y,-1)$, $x,y = \pm 1$. Similarly, the problematic sweep directions for the vertices on the bottom rough boundary are $\vec \omega = (x,y,1)$, $x,y = \pm 1$. Next, consider the vertices on the right rough boundary (see \cref{subfig:roughright} for an example). The problematic sweep directions for these vertices are also those that point inwards, i.e.\ $\vec \omega = (-1,y,z)$, $y,z = \pm 1$. Analogously, the problematic sweep directions for the vertices on the left rough boundary are $\vec \omega = (1,y,z)$, $y,z = \pm 1$. Therefore, the satisfying sweep directions for vertices on the rough boundaries are those that point outwards, as shown in  \Cref{fig:allowed_dirs}. 
    
    \begin{figure}
        \centering
        \subfloat[]{
            \centering
            \includegraphics[width=0.44\textwidth]{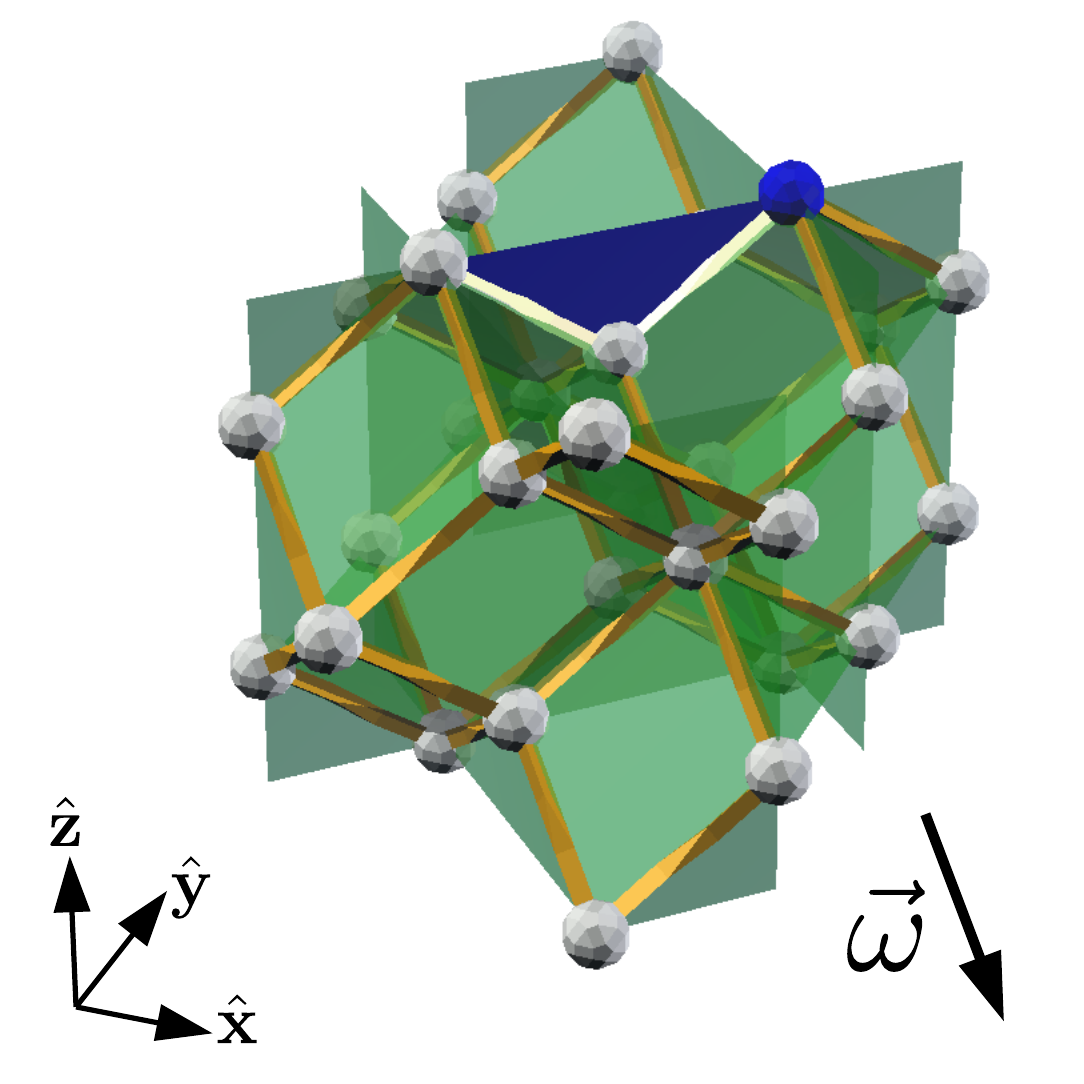}
            \label{subfig:roughtop}
        }
        \hfill
        \subfloat[]{
            \centering
            \includegraphics[width=0.44\textwidth]{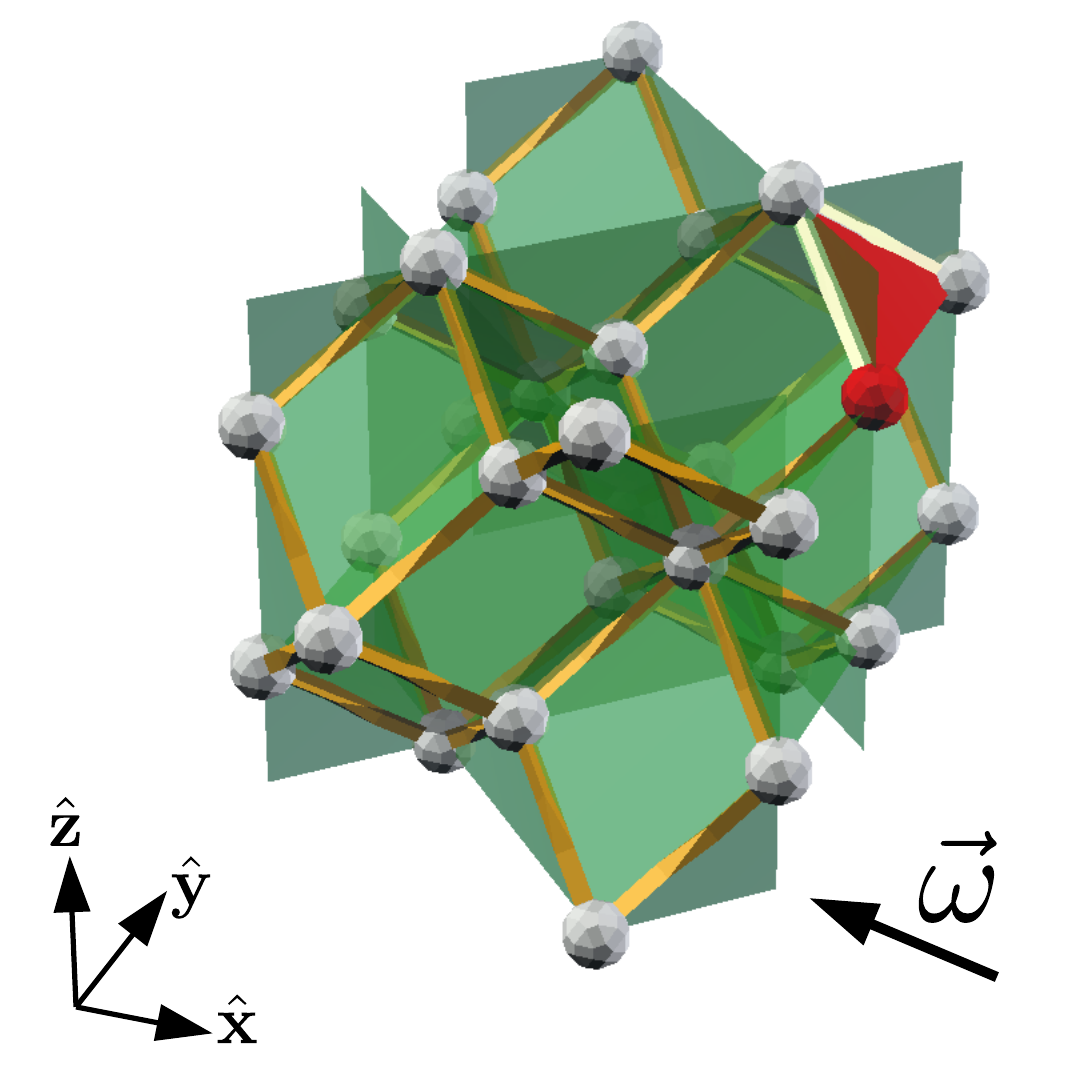}
            \label{subfig:roughright}
        }
        \caption{Vertices on the rough boundaries that do not satisfy the trailing vertex condition. In (a), we highlight such a vertex, $v$, in blue on the top rough boundary. Consider the highlighted qubit (blue face) and its syndrome (light yellow edges). For the sweep directions $\vec \omega = (1,-1,-1)$ and $\vec \omega' = (-1,1,-1)$, $v$ does not satisfy the trailing vertex condition because the blue face is not in $\future v$. In (b), we highlight a vertex  in red on the right rough boundary that does not satisfy the trailing vertex condition for the same reason as (a), where the problematic sweep direction is $\vec \omega = (-1, -1, 1)$.}
        \label{fig:rough}
    \end{figure}
    
    Now, consider the smooth boundaries. We have already analysed the vertices which are part of a rough boundary and a smooth boundary. For certain sweep directions, some vertices in the bulk of the smooth boundary do not satisfy the trailing vertex condition because of a missing face (see \cref{subfig:smooth} for an example). For each such vertex, there are two problematic directions (both of which point outwards from the relevant smooth boundary). In \cref{subfig:smooth}, these directions are $\vec \omega = -(1,1,1)$ and $\vec \omega' = (1,-1,1)$. However, for each smooth boundary there are four sweep directions (the ones that point inwards) for which every vertex in the bulk of the boundary satisfies the trailing vertex condition. \Cref{fig:allowed_dirs} illustrates the satisfying directions for each smooth boundary. 
    
    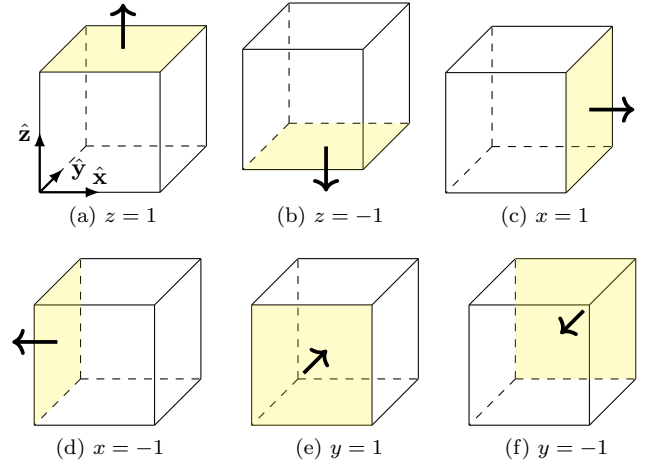
\begin{figure}
        \centering
        \centering
        \subfloat[$z=1$]{
            \centering
            \begin{tikzpicture}[scale=0.4]
                \coordinate[] (A) at (0,0,0);
                \coordinate[] (B) at (0,0,4);
                \coordinate[] (C) at (0,4,0);
                \coordinate[] (D) at (0,4,4);
                \coordinate[] (E) at (4,0,0);
                \coordinate[] (F) at (4,0,4);
                \coordinate[] (G) at (4,4,0);
                \coordinate[] (H) at (4,4,4);
                \draw [dashed] (A) -- (B);
                \draw [dashed] (A) -- (C);
                \draw [dashed] (A) -- (E);
                \draw (B) -- (D);
                \draw (B) -- (F);
                \draw (C) -- (D);
                \draw (C) -- (G);
                \draw (D) -- (H);
                \draw (E) -- (F);
                \draw (E) -- (G);
                \draw (F) -- (H);
                \draw (G) -- (H);
                \fill[yellow, opacity=0.2] (C) -- (D) -- (H) -- (G) -- cycle;
                \coordinate[] (Source) at (2, 4, 2);
                \coordinate[] (Drain) at (2, 5.5, 2);
                \draw [->, line width=0.5mm] (Source) -- (Drain);
                \draw [-latex, line width=0.3mm] (B) -- (0,0,2);
                \node[right] at (0,0,2) {$\hat{\mathbf y}$};
                \draw [-latex, line width=0.3mm] (B) -- (2,0,4);
                \node[above] at (2,0,4) {$\hat{\mathbf x}$};
                \draw [-latex, line width=0.3mm] (B) -- (0,2,4);
                \node[left] at (0,2,4) {$\hat{\mathbf z}$};
            \end{tikzpicture}
            \label{subfig:rt}
        }
        \hfill
        \subfloat[$z=-1$]{
            \centering
            \begin{tikzpicture}[scale=0.4]
                \coordinate[] (A) at (0,0,0);
                \coordinate[] (B) at (0,0,4);
                \coordinate[] (C) at (0,4,0);
                \coordinate[] (D) at (0,4,4);
                \coordinate[] (E) at (4,0,0);
                \coordinate[] (F) at (4,0,4);
                \coordinate[] (G) at (4,4,0);
                \coordinate[] (H) at (4,4,4);
                \draw [dashed] (A) -- (B);
                \draw [dashed] (A) -- (C);
                \draw [dashed] (A) -- (E);
                \draw (B) -- (D);
                \draw (B) -- (F);
                \draw (C) -- (D);
                \draw (C) -- (G);
                \draw (D) -- (H);
                \draw (E) -- (F);
                \draw (E) -- (G);
                \draw (F) -- (H);
                \draw (G) -- (H);
                \fill[yellow, opacity=0.2, dashed] (A) -- (B) -- (F) -- (E) -- cycle;
                \coordinate[] (Source) at (2, 0, 2);
                \coordinate[] (Drain) at (2, -1.5, 2);
                \draw [->, line width=0.5mm] (Source) -- (Drain);
            \end{tikzpicture}
            \label{subfig:rb}
        }
        \hfill
        \subfloat[$x=1$]{
            \centering
            \begin{tikzpicture}[scale=0.4]
                \coordinate[] (A) at (0,0,0);
                \coordinate[] (B) at (0,0,4);
                \coordinate[] (C) at (0,4,0);
                \coordinate[] (D) at (0,4,4);
                \coordinate[] (E) at (4,0,0);
                \coordinate[] (F) at (4,0,4);
                \coordinate[] (G) at (4,4,0);
                \coordinate[] (H) at (4,4,4);
                \draw [dashed] (A) -- (B);
                \draw [dashed] (A) -- (C);
                \draw [dashed] (A) -- (E);
                \draw (B) -- (D);
                \draw (B) -- (F);
                \draw (C) -- (D);
                \draw (C) -- (G);
                \draw (D) -- (H);
                \draw (E) -- (F);
                \draw (E) -- (G);
                \draw (F) -- (H);
                \draw (G) -- (H);
                \fill[yellow, opacity=0.2] (E) -- (F) -- (H) -- (G) -- cycle;
                \coordinate[] (Source) at (4, 2, 2);
                \coordinate[] (Drain) at (5.5, 2, 2);
                \draw [->, line width=0.5mm] (Source) -- (Drain);
            \end{tikzpicture}
            \label{subfig:rr}
        }
        \hfill
        \subfloat[$x=-1$]{
            \centering
            \begin{tikzpicture}[scale=0.4]
                \coordinate[] (A) at (0,0,0);
                \coordinate[] (B) at (0,0,4);
                \coordinate[] (C) at (0,4,0);
                \coordinate[] (D) at (0,4,4);
                \coordinate[] (E) at (4,0,0);
                \coordinate[] (F) at (4,0,4);
                \coordinate[] (G) at (4,4,0);
                \coordinate[] (H) at (4,4,4);
                \draw [dashed] (A) -- (B);
                \draw [dashed] (A) -- (C);
                \draw [dashed] (A) -- (E);
                \draw (B) -- (D);
                \draw (B) -- (F);
                \draw (C) -- (D);
                \draw (C) -- (G);
                \draw (D) -- (H);
                \draw (E) -- (F);
                \draw (E) -- (G);
                \draw (F) -- (H);
                \draw (G) -- (H);
                \fill[yellow, opacity=0.2, dashed] (A) -- (B) -- (D) -- (C) -- cycle;
                \coordinate[] (Source) at (0, 2, 2);
                \coordinate[] (Drain) at (-1.5, 2, 2);
                \draw [->, line width=0.5mm] (Source) -- (Drain);
            \end{tikzpicture}
            \label{subfig:rl}
        }
        \hfill
        \subfloat[$y=1$]{
            \centering
            \begin{tikzpicture}[scale=0.4]
                \coordinate[] (A) at (0,0,0);
                \coordinate[] (B) at (0,0,4);
                \coordinate[] (C) at (0,4,0);
                \coordinate[] (D) at (0,4,4);
                \coordinate[] (E) at (4,0,0);
                \coordinate[] (F) at (4,0,4);
                \coordinate[] (G) at (4,4,0);
                \coordinate[] (H) at (4,4,4);
                \draw [dashed] (A) -- (B);
                \draw [dashed] (A) -- (C);
                \draw [dashed] (A) -- (E);
                \draw (B) -- (D);
                \draw (B) -- (F);
                \draw (C) -- (D);
                \draw (C) -- (G);
                \draw (D) -- (H);
                \draw (E) -- (F);
                \draw (E) -- (G);
                \draw (F) -- (H);
                \draw (G) -- (H);
                \fill[yellow, opacity=0.2, dashed] (B) -- (D) -- (H) -- (F) -- cycle;
                \coordinate[] (Source) at (2.5, 2.5, 4);
                \coordinate[] (Drain) at (2.5, 2.5, 6);
                \draw [->, line width=0.5mm] (Drain) -- (Source);
            \end{tikzpicture}
            \label{subfig:sf}
        }
        \hfill
        \subfloat[$y=-1$]{
            \centering
            \begin{tikzpicture}[scale=0.4]
                \coordinate[] (A) at (0,0,0);
                \coordinate[] (B) at (0,0,4);
                \coordinate[] (C) at (0,4,0);
                \coordinate[] (D) at (0,4,4);
                \coordinate[] (E) at (4,0,0);
                \coordinate[] (F) at (4,0,4);
                \coordinate[] (G) at (4,4,0);
                \coordinate[] (H) at (4,4,4);
                \draw [dashed] (A) -- (B);
                \draw [dashed] (A) -- (C);
                \draw [dashed] (A) -- (E);
                \draw (B) -- (D);
                \draw (B) -- (F);
                \draw (C) -- (D);
                \draw (C) -- (G);
                \draw (D) -- (H);
                \draw (E) -- (F);
                \draw (E) -- (G);
                \draw (F) -- (H);
                \draw (G) -- (H);
                \fill[yellow, opacity=0.2, dashed] (A) -- (C) -- (G) -- (E) -- cycle;
                \coordinate[] (Source) at (1.5, 1.5, 0);
                \coordinate[] (Drain) at (1.5, 1.5, -2);
                \draw [->, line width=0.5mm] (Drain) -- (Source);
            \end{tikzpicture}
            \label{subfig:sb}
        }
        \caption{The constraints on the sweep directions $\vec \omega = (x,y,z) \in \{ (\pm 1, \pm 1, \pm 1) \}$ such that all vertices on the highlighted boundary satisfy the trailing vertex condition.}
        \label{fig:allowed_dirs}
    \end{figure}
    
    We recall that a local region of $\mathcal L$ has diameter smaller than $L/2$. Any such region can intersect at most two rough boundaries and one smooth boundary e.g.\ the boundaries highlighted in \cref{subfig:rt,subfig:rr,subfig:sf}. For this combination of boundaries, there is only one sweep direction for which all vertices satisfy the trailing vertex condition: $\vec \omega = (1,1,1)$. This sweep direction will also work for a local region that intersects any two of the above boundaries, and any region that intersects any one of the above boundaries. There are eight possible combinations of three boundaries that can be intersected by a local region (corresponding to the eight corners of the cubes in \cref{fig:allowed_dirs}). Therefore, by symmetry, the set of sweep directions required for the Lemma to hold are exactly $\Omega = \{ (\pm 1, \pm 1, \pm 1) \}$.
\end{proof}

We now explain how to modify the proof of \cref{lem:sr-props} such that it applies to our family of rhombic dodecahedral lattices with boundaries. First, we note that the proof of the support property is identical, except we replace $\dmond \sigma$ by $\reg \sigma$. The proof of the propagation property is essentially the same, except that there may be some time steps where the syndrome does not move, as there are some sweep directions for which certain syndromes are immobile. However, this does not affect the upper bound on the propagation distance of any syndrome: it is still upper-bounded by $T$, the number of applications of the rule. We note that in \cref{lem:sr-props}, the propagation distance refers to path lengths in the syndrome graph. However, one can verify that there will always be an edge linking any syndrome $\sigma^{(T)}$ and its corresponding syndrome at the following time step $\sigma^{(T+1)}$ in $\mathcal L$. Therefore, the upper bound on the propagation distance also applies to the distance defined as the length of the shortest path between two vertices of $\mathcal L$. 

Finally, the removal property also holds for any syndrome $\sigma$ with $\reg \sigma$ contained in a local region, but with a longer required removal time. Suppose that we use the ordering $\{ \vec \omega_1, \vec \omega_2, \ldots, \vec \omega_8 \}$ and apply the sweep rule for $T^{*}$ time steps using $\vec \omega_1$, followed by $T^{*}$ applications using $\vec \omega_2$, and so on until we reach $\vec \omega_8$, where 
\begin{equation}
    T^{*} = \max_{\vec \omega \in \Omega} \max_{\cpath{\inf \reg \sigma}{\sup \reg \sigma}} |\cpath{\inf \reg \sigma}{\sup \reg \sigma}|,
    \label{eq:T*-bound}
\end{equation} 
i.e.\ the longest causal path between the infimum and supremum of $\reg \sigma$, maximized over the set of sweep directions. 

By \cref{lem:trailing}, there will always exist a sweep direction such that the trailing location condition is satisfied at every vertex in $\reg \sigma$. Therefore, one can make a similar argument to the one made in the proof of \cref{lem:sr-props} to show that there exists a monotone for at least one sweep direction, $\vec \omega$, 
\begin{equation}
    f(T)=\max_{v \in \sigma^{(T)}}|\cpath{v}{\sup \sigma^{(T)}}|,
    \label{eq:monotone-bound}
\end{equation}
which decreases by one at each time step the sweep rule is applied with direction $\vec \omega$. By the support property $f_{\sigma}(T) \leq T^{*}$ for all $T$, so if we apply the rule $T^{*}$ times with sweep direction $\vec \omega$, then the syndrome is guaranteed to be removed. But a priori we do not know which sweep direction(s) will remove a given syndrome. Therefore, to guarantee the removal of a syndrome, we must apply the sweep rule $T^{*}$ times in each direction, which gives a total removal time of $|\Omega|\times T^{*}$. 

\section*{Data availability}

The code used in this study is available in the GitHub repository, \url{https://github.com/MikeVasmer/Sweep-Decoder-Boundaries}. The data are available from the corresponding author upon reasonable request. 

\section*{Acknowledgements}

M.V. thanks Niko Breuckmann, Padraic Calpin, Earl Campbell, and Tom Stace for valuable discussions.
A.K. acknowledges funding provided by the Simons Foundation through the ``It from Qubit'' Collaboration.
This research was supported in part by Perimeter Institute for Theoretical Physics.
Research at Perimeter Institute is supported in part by the Government of Canada through the Department of Innovation, Science and Economic Development Canada and by the Province of Ontario through the Ministry of Colleges and Universities.
Whilst at UCL, M.V. was supported by the EPSRC (grant number EP/L015242/1).
D.E.B. acknowledges funding provided by the EPSRC (grant EP/R043647/1).
We acknowledge use of \href{https://www.ucl.ac.uk/research-it-services/services/research-computing-platforms}{UCL supercomputing facilities}.

\section*{Author contributions}

Concept conceived by all the authors. A.K. generalized the sweep rule. All authors contributed to the design of the sweep decoder. Analytic proofs by A.K. and M.V. Simulations designed and written by M.V., with input from D.E.B. and A.K. Data collected and analysed by M.V. Manuscript prepared by M.V. and A.K.

\section*{Competing Interests}

The authors declare that there are no competing interests.

\bibliographystyle{utphys}
\bibliography{references}

\onecolumngrid
\appendix

\section*{Supplementary Note 1 \label{app:thresh}}

In this note, we fill in the details of the proof of the non-zero threshold of the sweep decoder (Theorem 1). The proof strategy is essentially the same as in Ref.~\cite{Kubica2019}. First, we describe a decomposition of errors into chunks, which aids our analysis. Then, we show that the sweep decoder successfully corrects chunks of the error, up to a certain size. Finally, we find that the probability of the error containing chunks that cause the decoder to fail is exponentially suppressed in the linear size of the lattice. 

\subsection*{Chunk decomposition of the error}
\label{subsec:chunks}

We begin by defining the diameter of a subset of vertices, $U\subseteq \mathcal L_0$, to be the maximal distance between any two vertices in $U$, i.e.\ $\diam U = \max_{u,v \in U} \dist{u}{v}$, where $\dist{u}{v}$ is the length of the shortest path linking $u$ and $v$ in $\mathcal L$. Let $\epsilon \subset \mathcal L_2$ denote an error contained in a local region of $\mathcal L$. We now define a decomposition of errors into chunks, following~\cite{Bravyi2013}. A level-0 chunk, $E^{[0]}$ is an element of $\epsilon$ and a level-$n$ chunk is a defined recursively to be the disjoint union of two level-$(n-1)$ chunks $E_{1}^{[n-1]}$ and $E_{2}^{[n-1]}$, where $\diam{E^{[n]}}\leq Q^{n}/2$ for some constant $Q$. The level-$n$ error $E_{n}\subseteq\epsilon$ is defined to be the union of all level-$n$ chunks. We note that 
\begin{equation}
    \epsilon = E_0 \supseteq \ldots \supseteq E_m \supsetneq E_{m+1} = \emptyset.
    \label{eq:E-inclusion}
\end{equation}
Therefore, we can decompose the error into $F_{n}=E_{n}\setminus E_{n+1}$, for $n \in \{0, \ldots, m\}$ as follows
\begin{equation}
    \epsilon = F_0 \sqcup \ldots \sqcup F_m,
    \label{eq:F-inclusion}
\end{equation}
where $A \sqcup B$ denotes the disjoint union of $A$ and $B$. We say that $M\subseteq\epsilon$ is an $l$-connected component of $\epsilon$ if, for any $M_{1}, M_{2} \neq \emptyset$, if $M = M_1 \sqcup M_2$ then $d(M_1, M_2) \leq l$. We need the following lemma, which concerns the size and separation of connected components of $F_{i}$. 
\begin{lemma}{(Connected Components~\cite{Bravyi2013})}.
Let $\epsilon$ be an error with disjoint decomposition $\epsilon=F_{0} \sqcup \ldots \sqcup F_{m}$ and let $Q\geq 6$ be a constant. Suppose that $M\subseteq\epsilon$ is a $Q^{n}$-connected component of $F_{n}$. Then, $\diam{M}\leq Q^{n}$ and $\dist{M}{E_{n}\setminus M}>Q^{n+1}/3$.
\label{lem:concom}
\end{lemma}
\cref{lem:concom} gives us both an upper bound on the size of any $Q^{n}$-connected component of the error and a lower bound on the separation of the $Q^{n}$-connected component from the rest of the error. For a proof, see~\cite{Bravyi2013,Kubica2019}. 

\subsection*{Correction of high-level chunks}
\label{subsec:chunk-rm}

\begin{lemma}
    Let $\epsilon \subset \mathcal L_2$ be an error with disjoint decomposition $\epsilon=F_{0}\sqcup F_{1} \ldots F_{m^*-1}$. Choose constants $Q = 6 |\Omega| c_{P} c_{D}^8$ and $m^{*}=\ceil{\log_{Q}(L/2 c_{D}^8)}$, where $L$ is the linear lattice size, $\Omega$ is the set of sweep directions, $c_{D}=2$ and $c_{P}=1$. Suppose we apply the sweep decoder with $T_{max}=c_{P}c_{D}^8 Q^{m^*}$. Then, $\epsilon$ is corrected, i.e.\ the product of $\epsilon$ and the correction returned by the sweep decoder is a stabilizer.
\label{lem:chunkremoval}
\end{lemma}

The proof of the above Lemma is essentially the same as the proof of Lemma~3 in~\cite{Kubica2019}, albeit with adjusted constants. We briefly sketch the proof here, for completeness.

Let $M$ be a $Q^n$-connected component of $F_n$. We chose the constant $m^* = \ceil{\log_Q (L/2 c_D^8)}$ such that $\reg M$ is contained in a local region of $\mathcal L$. One can verify that $\diam{\reg U} \leq c_D \diam U$ for all $U \subseteq \mathcal L_0$ and $\vec \omega \in \Omega$, with $c_D = 2$. Therefore 
\begin{equation}
    \diam{\reg M} 
    \leq c_D^8 \diam M
    \leq c_D^8 Q^n 
    < c_D^8 Q^{m^*} \leq \frac{L}{2}
\end{equation}

By the removal property of the sweep rule, $\sigma = \partial M$ will be removed in $|\Omega| \times T^*$ time steps, where $T^*$ is 
\begin{equation}
    T^{*} = \max_{\vec \omega \in \Omega} \max_{\cpath{\inf \reg \sigma}{\sup \reg \sigma}} |\cpath{\inf \reg \sigma}{\sup \reg \sigma}|.
    \label{eq:T*-bound-app}
\end{equation}
We use the properties of the rhombic dodecahedral lattice to bound $T^*$. For any two vertices $u,v \in \mathcal L_0$, and any sweep direction $\vec \omega \in \Omega$, $\max_{\cpath{u}{v}}|\cpath{u}{v}|\leq c_{P}\times\dist{u}{v}$, with $c_P=1$. Therefore, we can upper-bound $T^*$:
\begin{equation}
    T^* \leq |\Omega| c_P \diam{\reg M} 
    \leq c_P c_D^8 Q^n.
\end{equation}

We choose the constant $Q = 6 |\Omega| c_P c_D^8$ such that $M$ is removed independently from $E_i \setminus M$. This follows from \cref{lem:concom} and the propagation property of the rule. Finally, because $\reg M$ is contained in a local region of $\mathcal L$, the product of $M$ and its correction implements a trivial logical operator. Given the above, an inductive argument shows that all chunks up to level-$m^*$ are corrected by the sweep decoder. 

\subsection*{Probability of high-level chunks}
\label{subsec:chunk-supp}

The only remaining step in the proof of Theorem 1 is to show that the probability of the error $\epsilon$ containing level-$m^{*}$ chunks is exponentially suppressed in the size of the lattice. This can be accomplished using a percolation theory argument, as explained in~\cite{Kubica2019}. The outcome is the following bound on the probability of $\epsilon$ containing a level-$n$ (or higher) chunk:
\begin{equation}
    \Pr[\text{$\epsilon$ contains a level-$n$ chunk}]\leq \frac{|\mathcal L_0|}{\lambda^{2}} \left(\frac{p}{p_{\mathrm{th}}}\right)^{2^{n}},
    \label{eq:mstar-chunk-pr}
\end{equation}
where $\lambda=(2Q)^{3}c_{B}$. To specify $c_B$, we define a discrete ball of radius $r$ centred at a vertex $v$, $B_v(r)$, to be the set of all lattice elements a distance smaller than $r$ from $v$. The constant $c_B=8$ is set by the fact that for any ball $B_{v}(R)$ contained within a finite region of $\mathcal L^\infty$, there exists a cover
\begin{equation}
    \bigcup_{u\in U}B_{u}(r)\supset B_{v}(R),
    \label{eq:le1}
\end{equation}
consisting of balls of radius $r<R$ index by $U\subset \mathcal L^\infty$, such that 
\begin{equation}
    |U|\leq c_{B}(R/r)^{3}.
    \label{eq:le2}
\end{equation}   
The threshold error probability $p_{\mathrm{th}}$ is 
\begin{equation}
    p_{\mathrm{th}}=\left(\lambda^{2}\max_{v\in\tri{0}{\mathcal{L}}}|\St_{2}(v)|\right)^{-1},
    \label{eq:pth}
\end{equation}
where the $2$-star of $v$, $\St_{2}(v)=\{ f \in \mathcal L_2 : v \in f \}$. If we substitute $m^{*}=\ceil{\log_{Q}(L/2c_{D}^8)}$ into \cref{eq:mstar-chunk-pr}, we obtain:
\begin{equation}
    \Pr[\text{$\epsilon$ contains a level-$m^{*}$ chunk}]\leq \frac{|\mathcal L_0|}{\lambda^{2}} \left(\frac{p}{p_{\mathrm{th}}}\right)^{\beta_1 L^{\beta_2}},
    \label{eq:mstar-chunk-pr-L}
\end{equation}
where $\beta_1=1/(2 c_D^8)^{\beta_1}$ and $\beta_2=\log_{Q}2$. As $|\mathcal L_0| = \BigO{ L^3 }$, the probability of the error containing a level-$m^{*}$ (or higher) chunk is $\BigO{(p/p_{\mathrm{th}})^{\beta_1 L^{\beta_2}}}$. \qed

For rhombic dodecahedral lattices, the value of the error threshold given by \cref{eq:pth} is $p_{\mathrm{th}}\approx 10^{-30}$. This is many orders of magnitude smaller than the value we observe in simulations of $p_{\mathrm{th}} \approx 21.5\% $, which underlies the importance of using numerical simulations to estimate the error threshold of a decoder. 

% \section{Simulation results for cubic lattice toric codes \label{app:cubic}}
\section*{Supplementary Note 2 \label{app:cubic}}

We applied the sweep decoder to toric codes defined on cubic lattices with open and periodic boundary conditions. For the case with open boundaries, we consider a family of toric codes constructed from the infinite cubic lattice. To construct a toric code with code distance $L$ and one encoded qubit, we take a $L \times L \times (L-1)$ sublattice of the infinite cubic lattice with vertices at integer coordinates $(x,y,z) \in [0, L] \times [0, L] \times [0, L-1]$. Then we associate qubits ($X$ stabilizer generators) with all faces (edges) except those in the $x=0,L$ or $y=0,L$ planes. Supplementary Figure~\ref{fig:cubic-L} illustrates the construction of the $L=3$ lattice. We use eight sweep directions $\Omega = \{(\pm 1, \pm 1, \pm 1)\}$, where each sweep direction points into the centre of the cubes. Supplementary Figure~\ref{subfig:cubic-sr} explains how the rule works for one of the sweep directions (the rest are the same by symmetry). 

We simulated the performance of the decoder for an error model with equal phase-flip and measurement error probabilities (i.e.\ $\alpha=1$ in Definition 2). We observe a sustainable error threshold of $p_{\mathrm{sus}}\approx 1.7\%$ for toric codes defined on lattices with and without boundaries; see Supplementary Figure~\ref{fig:cubic-sus}. We find that the optimal decoder parameters were essentially the same as those we described for the rhombic dodecahedral lattice (see Results of the main text). 

\begin{figure}
    \centering
    \subfloat[]{
        \centering
        \includegraphics[width=0.3\textwidth]{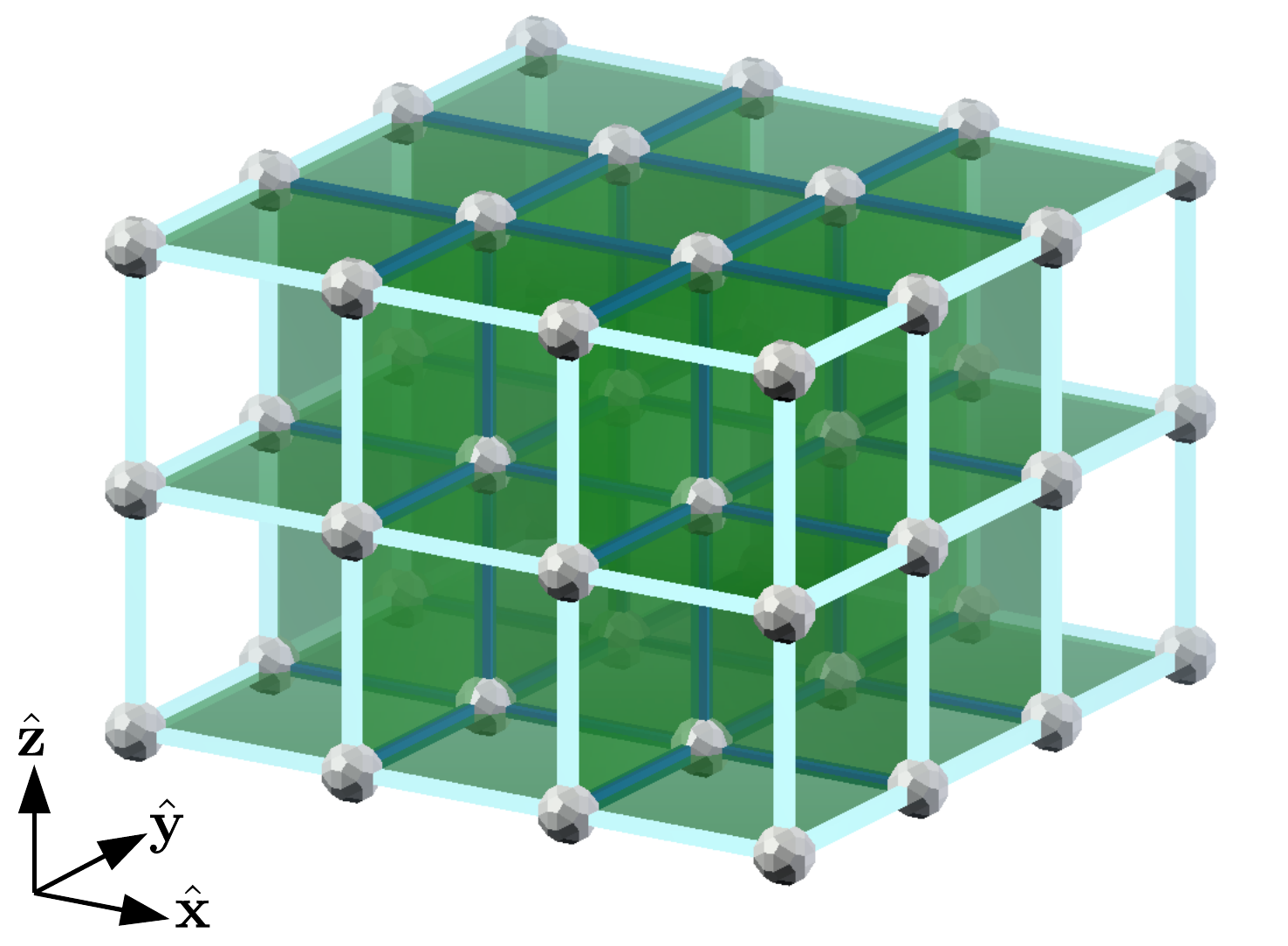}
    }
    \hfill
    \subfloat[]{
        \centering
        \includegraphics[width=0.3\textwidth]{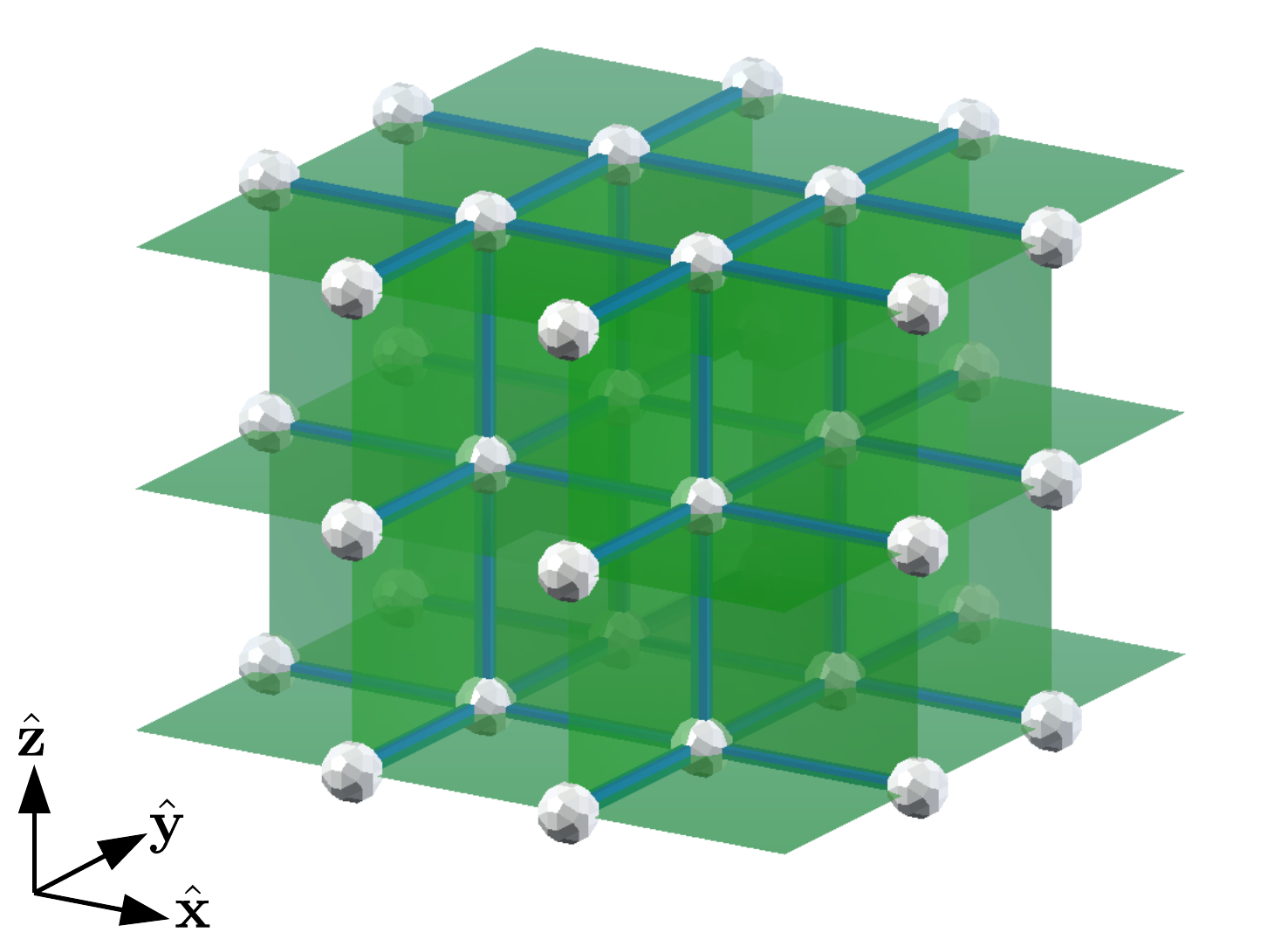}
    }
    \hfill
    \subfloat[]{
        \label{subfig:cubic-sr}
        \centering
        \includegraphics[width=0.2\textwidth]{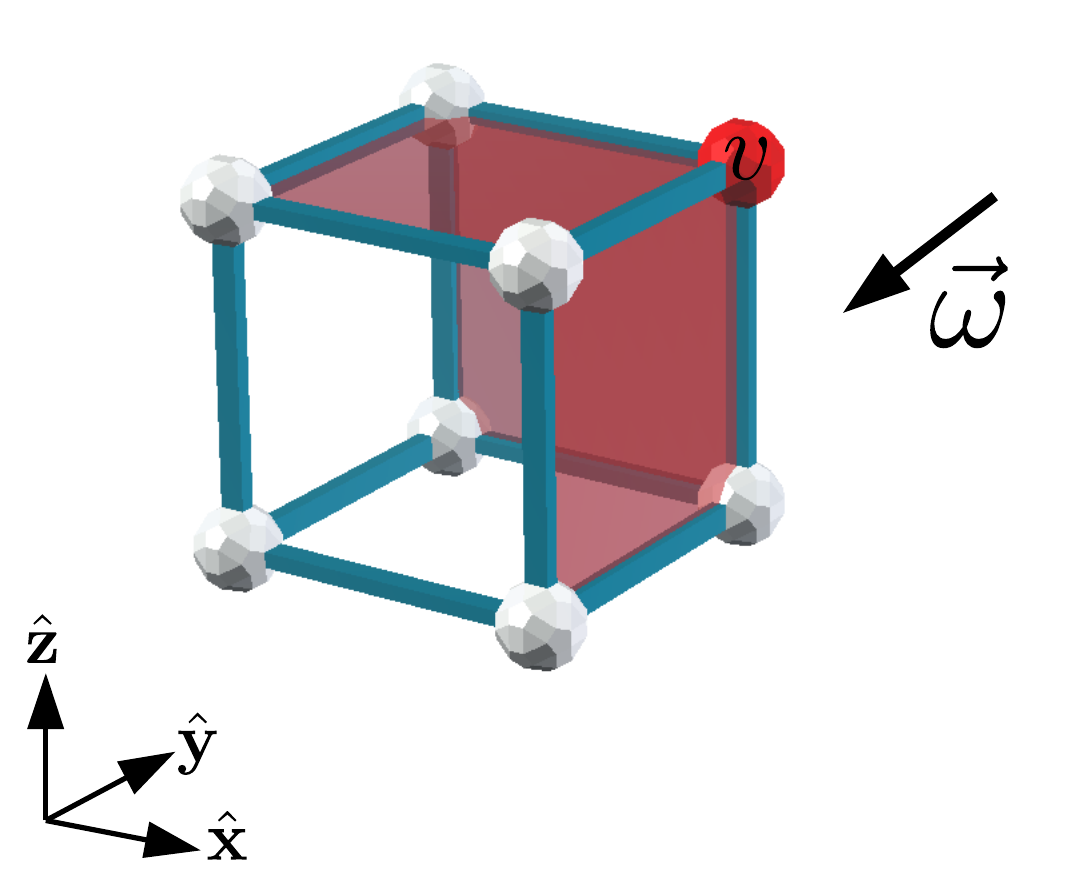}
    }
    \caption{(a) Constructing the $L=3$ cubic toric code lattice. We illustrate how the lattice (dark blue edges and green faces) is a sublattice of the infinite cubic lattice (represented by light blue edges). The $L=3$ cubic toric code lattice. We show all the vertices, edges and faces of the lattice. In particular, we note that some faces on the boundary have only two or three edges. (c) The action of the sweep rule at the vertices of the cubic lattice. The sweep direction $\vec \omega = -(1,1,1)$ points from the red vertex $v$ into the centre of the cube. The shaded red faces are the faces that can be returned by the rule from $v$.}
    \label{fig:cubic-L}
\end{figure}

\begin{figure}
    \centering
    \includegraphics[width=0.45\textwidth]{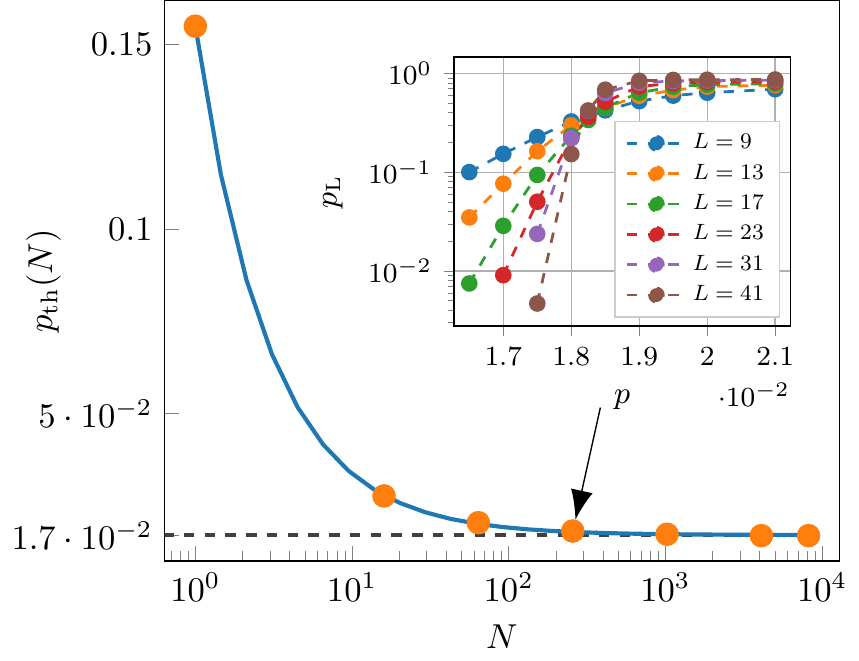}
    \caption{Sustainable threshold of the sweep decoder applied to the toric code on the cubic lattice (with boundaries). We plot $p_{\mathrm{th}}(N)$, the error threshold as a function of the number $N$ of error correction cycles, for an error model with $\alpha = q/p = 1$. The inset shows the data for $N=2^{10}$, where we use $10^4$ Monte Carlo samples for each point. We observe a sustainable threshold of $p_{\mathrm{sus}}\approx 1.7\%$, with $\gamma = 0.92$ (see Eq. (10) for the fit).}
    \label{fig:cubic-sus}
\end{figure}

% \section{Performance of the sweep decoder against correlated noise \label{app:corr}}
\section*{Supplementary Note 3 \label{app:corr}}

We ran simulations to estimate the performance of the sweep decoder against correlated noise for rhombic dodecahedral lattices with boundaries. We used a simple error model where, at each time step, every pair of neighbouring qubits experiences an error with probability $p$, where the error is drawn randomly from $\{ZI,IZ,ZZ\}$. In addition, each stabilizer measurement outcome is flipped with probability $q=p$. To compare the performance of the decoder for this error model against the iid phase-flip error model described in Definition 2, we cannot simply use the same values of $p$, as the parameter has different meanings in each model. Instead, we use an effective error rate~\cite{Maskara2019} $p_{eff}$, which is the marginal probability that a given qubit experiences a phase-flip. In the iid phase-flip error model, $p_{eff}=p$, but in the correlated error model $p_{eff}=2p-8p^{2}/3+\BigO{p^{3}}$. Using the effective error rate as our parameter, we find that behaviour of $p_{\mathrm{th}}(N)$ for correlated noise is analogous to the iid phase-flip case, except with a lower sustainable threshold of $p_{\mathrm{sus}} \approx 0.8\%$. Supplementary Figure~\ref{fig:corr-sus} shows the data.

\begin{figure}
    \centering
    \includegraphics[width=0.45\textwidth]{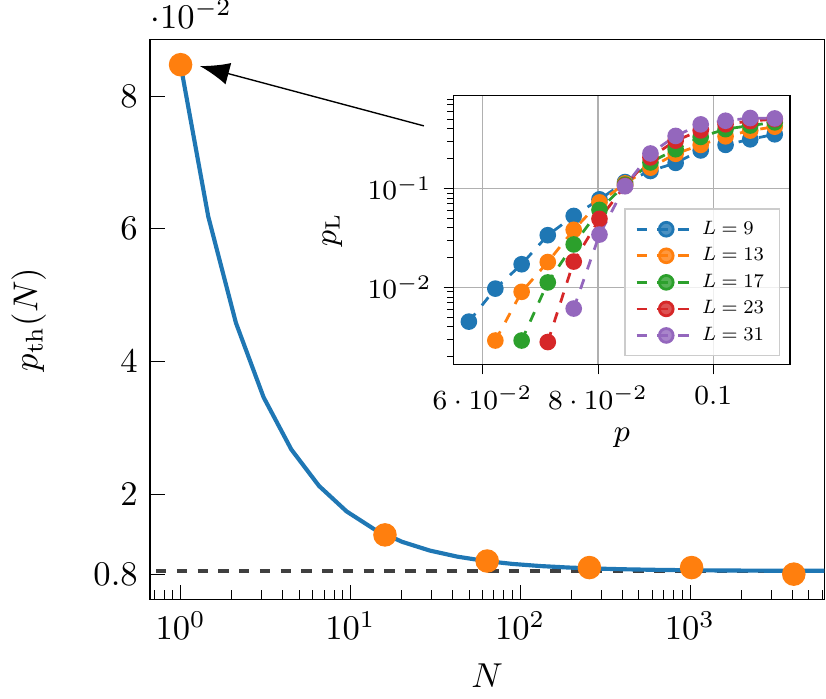}
    \caption{Sustainable threshold of the sweep decoder applied to the toric code on the rhombic dodecahedral lattice (with boundaries) against correlated errors. We plot $p_{\mathrm{th}}(N)$, the error threshold as a function of the number of error correction cycles. The inset shows the data for $N=1$, where we use $10^4$ Monte Carlo samples for each point. We observe a sustainable threshold of $p_{\mathrm{sus}}\approx 0.8\%$, with $\gamma = 0.95$ (see Eq. (10) for the fit).}
    \label{fig:corr-sus}
\end{figure}

% \clearpage

\end{document}